\documentclass[12pt,formal]{article}

\usepackage{etex}
\usepackage{algorithmic}
\usepackage{algorithm}

\usepackage{fullwidth}
\usepackage{etoolbox}
\usepackage{ifthen}
\usepackage{float}
\usepackage{tikz}
\usepackage{multicol}
\usepackage{todonotes}
\usepackage{enumitem}
\usepackage{wrapfig}
\usepackage{adjustbox}
\usepackage{tkz-graph}
\usetikzlibrary{arrows,automata}
\usetikzlibrary{arrows,shapes,positioning,scopes}
\usetikzlibrary{tikzmark,decorations.pathreplacing,calc,fit}
\usetikzlibrary{arrows,automata}
\usetikzlibrary{arrows,petri}
\GraphInit[vstyle = Shade]
\usepackage{ulem}
\usepackage[latin1]{inputenc}
\usepackage{ifthen}
\usepackage{mathrsfs}
\usepackage{multirow}
\usepackage{rotating}
\usetikzlibrary{matrix,arrows}
\usepackage{bussproofs}
\usepackage{qtree}
\usepackage{easyReview}
\usepackage[polutonikogreek,english]{babel}
\usepackage{tikz}
\usepackage{adjustbox}

\usetikzlibrary{arrows,petri}
\usetikzlibrary{arrows.meta}

\usepackage{pgf}
\usepackage{times}
\usepackage[T1]{fontenc}
\usepackage{amsfonts,amsmath,amssymb,amsthm}
\usepackage[all]{xy}

\usepackage{dingbat}
\usepackage{subcaption}
\usepackage{cmll}
\floatstyle{plain}
\floatstyle{ruled}
\restylefloat{algorithm}
\newfloat{myalgo}{tbhp}{mya}
\newcommand{\INDSTATE}[1][1]{\State\hspace{1cm}}
\newlength\myindent
\newtheorem{fact}{Fact}

\newcommand{\size}[1]{\ensuremath{\left|#1\right|}}

\newcommand{\imply}{\supset}

\newcommand{\mil}{\ensuremath{\mathbf{M}_{\imply}}}

\newcommand{\Branch}[1]{\overrightarrow{#1}}
\newcommand{\symbvdash}{'\vdash'}

\newboolean{numberspec}
\setboolean{numberspec}{false}
\newcounter{specline}

\newcommand{\spec}[3][{}]{%
   \setcounter{specline}{0}%
   \ifthenelse{\equal{#1}{numbers}}%
      {\setboolean{numberspec}{true}}%
      {\setboolean{numberspec}{false}}%
   \ensuremath{
      \begin{array}{l}
         \ifthenelse{\equal{#2}{}}
            {\numberedline}
            {#2\nl}
         #3
      \end{array}
   }
}

{\begin{myalgo}[#1]
\centering
\begin{minipage}{#2}
\begin{algorithm}[H]}%
{\end{algorithm}
\end{minipage}
\end{myalgo}}

\newbool{formal}
\booltrue{formal}

\input{tcilatex}

\begin{document}
%\begin{frontmatter}
  \title{Going from the huge to the small: Efficient succint representation of proofs in Minimal implicational logic}
  \author{Edward Hermann Haeusler$^\dagger$
           \\ Departmento de Inform\'{a}tica\\ PUC-Rio\\
              Rio de Janeiro, Brasil \\
              $^\dagger$Email:hermann@inf.puc-rio.br}

  \maketitle

   \begin{abstract}

     A previous article shows that any linear height bounded normal proof of a tautology in the Natural Deduction for Minimal implicational logic (\mil) is as huge as it is redundant. More precisely, any proof in a family of super-polynomially sized and linearly height bounded
proofs have a sub-derivation that occurs super-polynomially many times in it. In this article, we show that by collapsing all the repeated sub-derivations we obtain a smaller structure, a rooted Directed Acyclic Graph (r-DAG), that is polynomially upper-bounded on the size of $\alpha$ and it is a certificate that $\alpha$ is a tautology that can be verified in polynomial time. In other words, for every huge proof of a tautology in \mil, we obtain a succinct certificate for its validity. Moreover,  we show an algorithm able to check this validity in polynomial time on the certificate's size. Comments on how the results in this article are related to a proof of the conjecture $NP=CoNP$ appears in conclusion.

\end{abstract}

   \section{Introduction}

   In \cite{Exponential} and \cite{SuperPoly} we discuss the correlation between the size of proofs and how redundant they can be. A proof or logical derivation is redundant whenever it has sub-proofs that are repeated many times inside it. The articles \cite{Exponential} and \cite{SuperPoly} focus on Natural Deduction (ND) proofs in the purely implicational minimal logic $\mil$. This logic, here called $\mil$, is PSPACE-complete and simulates polynomially any proof in Intuitionistic Logic and full minimal logic, being hence an adequate representative to study questions regarding computational complexity. The fact that $\mil$ has straightforward syntax and ND system is worthy of notice. Moreover, compressing proofs in $\mil$ can provide very good glues to compress proofs in any one of these mentioned systems, even for the Classical Propositional Logic. In \cite{BoSL} and \cite{Studia2019}, we prove that for every \mil tautology $\alpha$ there is a two-fold certificate for the validity of $\alpha$ in \mil. The certificate is polynomially sized on the length of $\alpha$ and verifiable in polynomial time on this length too. This is the general approach of our proof that $NP=PSPACE$. It is well-known that $NP=PSPACE$ implies $CoNP=NP$, so we can conclude also that $CoNP=NP$. This article, together with the articles \cite{Exponential}, \cite{SuperPoly}, and one of the appendixes of \cite{Exponential},  aim to show an alternative and more intuitive proof that $NP=CoNP$. Recently, we deposited in arxiv a note, see \cite{Addendum},  that explains how to have a third alternative proof of $NP=CoNP$, this time with a double certificate on linear height normal proofs, without to use Hudelmaier result.   This paper aims to show how to use the inherent redundancy in huge proofs to have polynomial and polytime certificates by removing this redundancy from the proofs. This is made by collapsing all the redundant sub-proofs into only one occurrence in the proof, by type. We start with tree-like Natural Deduction proofs and end up with a labeled r-DAG (rooted Directed Acyclic Graph). In the sequel we explain and show what is said in the last two phrases, previously stated. Currently, the use of the redundancy theorem and corollary shown in \cite{SuperPoly}, that is the essence of this proof's approach, is not easily adaptable to a proof of $NP=PSPACE$. In \cite{BoSL} the linearly height  upper-bounded proofs of the tautologies in \mil are do not have have to be normal. 

   In \cite{Exponential}, we identify sets of huge proofs with sets of proofs that, when viewed as strings, have their length lower-bounded by some exponential function. Moreover, we can consider, without loss of generality,  proof/deductions, which are linearly height bounded, as stated in \cite{Studia2019}.  We prove that the exponentially lower-bounded $\mil$ proofs are redundant, in the sense that there is at least one sub-proof for each proof that occurs exponentially many times in it.
In \cite{SuperPoly},  we show that this result extends to super-polynomial proofs, i.e., proofs that are lower-bounded by any polynomial. We consider huge proofs/derivations as super-polynomial sized proofs.  Thus, we prove that, in any set of super-polynomially lower-bounded proofs in $\mil$, of tautologies, all proof are redundant. Redundancy means that there is a sub-proof that occurs super-polynomially many times in it for almost every proof in this set of huge proofs. 

In section~\ref{sec:Background}, we present a brief presentation and explanation on the proof-theoretical terminology used here and the results from our previous articles. For a more detailed and comprehensive reference, see \cite{Exponential}, \cite{SuperPoly}, \cite{BoSL} and \cite{Studia2019}. Section~\ref{sec:RemovingRedundancies} is the section where we show how to remove redundancies using a recursive method adequately. Section~\ref{sec:PolytimeCompression} shows that the certificate obtained in section~\ref{sec:RemovingRedundancies} is of polynomial size and can be checked in polynomial time too. We conclude this article in section~\ref{sec:Conclusion}. Finally, due to have a self-contained article, and the background and previous results that this paper uses, we advise the reader that there is an essential superposition of this article with \cite{SuperPoly} and \cite{Exponential}.

\section{Background on proof-theory and brief presentation of previous results}
\label{sec:Background}

One reason to study redundancy in proofs is to obtain a compressing method based on redundancy removals. According to
the redundancy theorem (theorem 12, page 17 in \cite{SuperPoly}) and its corollary (corollary 13, page 18) proofs belonging to a family of super-polynomial proofs are super-polynomial redundant. The reader can find all the content of this section in \cite{SuperPoly} in more detail.  This section contains excerpts of \cite{SuperPoly}. 

The Natural Deduction system,  defined in \cite{Gentzen1936}),  is taken as a set of inference rules that settle the concept of a logical deduction. Natural Deduction does not have axioms. It implements in the level of the logical calculus the  {\em (meta)theorem of deduction}, $\Gamma, A\vdash A\imply B$, via the discharging mechanism. The $\imply$-introduction rule, for example, uses this discharging mechanism in the logic calculus level. 

\begin{prooftree}
  \AxiomC{[$A$]}
  \noLine
  \UnaryInfC{$\Pi$}
  \noLine
  \UnaryInfC{$B$}
  \RightLabel{$\imply$-Intro}
  \UnaryInfC{$A\imply B$}
\end{prooftree}

We embrace formulas occurrence  $A$ in the derivation $\Pi$ of $B$ from $A$ with a pair of [] to indicate the discharge of them. An embraced formula occurrence [$A$] means that from the $\imply$-Intro rule discharging it down to the conclusion of the derivation, the inferred formulas do not depend anymore on this occurrence $A$.
The choice of formulas to be discharged in an $\imply$-Intro rule is arbitrary and liberal. The range of this choice goes from every occurrence of $A$, the discharged formula until none of them.

The following derivations show two different ways of deriving $A\imply (A\imply A)$. Observe that in both deductions or derivations, we use numbers to indicate the application of the $\imply$-Intro that discharged the marked formula occurrence. For example, in the right derivation, the upper application discharges the marked occurrences of $A$, while in the left derivation, it is the lowest application that discharges the formula occurrences $A$. There is a third derivation that both applications do note discharge any $A$, and the conclusion $A\imply(A\imply A)$ keep depending on $A$. This third alternative appears in figure~\ref{third}. Natural Deduction systems can provide logical calculi without any need to use axioms. In this article, we focus on the system formed only by the $\imply$-Intro rule and the $\imply$-Elim rule, as shown below, also known by {\em modus ponens}. The logic behind this logical calculus is the purely minimal implicational logic, $M_{\imply}$.

\begin{prooftree}
  \AxiomC{$A$}
  \AxiomC{$A\imply B$}
  \RightLabel{$\imply$-Elim}
  \BinaryInfC{$B$}
\end{prooftree}

Without loss of generality, we substitute the liberal discharging mechanism by a greedy one that discharges every possible formula occurrence whenever the $\imply$-Intro is applied. As stated and proved in \cite{SuperPoly}, if there is an N.D. proof of $\alpha$ in \mil, then there is a proof of $\alpha$ in \mil that has all applications of $\imply$-Intro as greedy ones.

\begin{prooftree}
  \AxiomC{$[A]^1$}
  \UnaryInfC{$A\imply A$}
  \RightLabel{$\;^1$}
  \UnaryInfC{$A\imply (A\imply A)$}
  \AxiomC{$[A]^1$}
  \RightLabel{$\;^1$}
  \UnaryInfC{$A\imply A$}
  \UnaryInfC{$A\imply (A\imply A)$}
  \noLine
  \BinaryInfC{}
\end{prooftree}

\begin{figure}[H]
\begin{prooftree}
  \AxiomC{$A$}
  \UnaryInfC{$A\imply A$}
  \UnaryInfC{$A\imply (A\imply A)$}
\end{prooftree}
\caption{Two vacuous $\imply$-Intro applications}\label{third}
\end{figure}

In our previous articles, we consider Natural Deduction as trees with the sake of having simpler proofs of our results. There is a binary tree with nodes labelled by the formulas and edges linking premises to the conclusion for any ND derivation. The tree's root is the conclusion of the derivation, and the leaves are its assumptions. Figure~\ref{ex:derivacao} has the tree in figure~\ref{ex:tree} representing it.  In a proof-tree, the set of formulas that the label of $u$ depends on $v$  labels the edge from  $v$ to  $u$. This set of formulas is called the dependency set of $u$ from $v$. The greedy version of the $\imply$-intro removes the discharged formula from the corresponding dependency sets, as shown in figure~\ref{ex:tree}. We need one more extra edge and the root node. The dependency set of the conclusion labels this edge. That is why the edge links the conclusion to the dot in figure~\ref{ex:tree}.

\begin{figure}[H]
  \begin{prooftree}
  \AxiomC{$[A]^{1}$}
  \AxiomC{$A\imply B$}
  \BinaryInfC{$B$}
  \AxiomC{$B\imply C$}
  \BinaryInfC{$C$}
  \UnaryInfC{$\LeftLabel{1}A\imply C$}
  \end{prooftree}
\caption{A derivation in $M_{\imply}$}\label{ex:derivacao}
\end{figure}

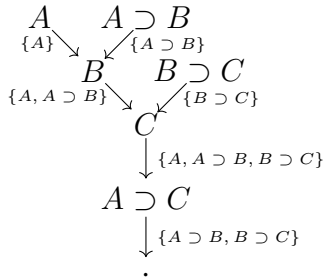
\begin{figure}[H]
      \begin{tikzpicture}[shorten >=1pt,node distance=1cm,auto]%,on grid
\tikzstyle{state}=[shape=circle,thick,minimum size=1.5cm]
\node[state] (ground) {.};
\node[state, above of=ground] (AC) {$A\imply C$};
\node[state, above of=AC] (C) {$C$};
\node[state, above left of=C] (B) {$B$};
\node[state, above left of=B] (A) {$A$};
\node[state, above right of=B] (AB) {$A\imply B$};
\node[state, above right of=C] (BC) {$B\imply C$};
\path [draw,->] (ground) edge node[midway] {{\tiny $\{A\imply B, B\imply C\}$}} (AC);
\path [draw,->] (AC) edge node[midway] {{\tiny $\{A,A\imply B, B\imply C\}$}} (C);
\path [draw,->] (C) edge node[midway,right] {{\tiny $\{B\imply C\}$}} (BC);
\path [draw,->] (C) edge node[midway, left] {{\tiny $\{A,A\imply B\}$}} (B);
\path [draw,->] (B) edge node[midway, right] {{\tiny $\{A\imply B\}$}} (AB);
\path [draw,->] (B) edge node[midway, left] {{\tiny $\{A\}$}} (A);
    \end{tikzpicture}
\caption{The tree representing the derivation in figure~\ref{ex:derivacao}}\label{ex:tree}
\end{figure}

Finally, we use bitstrings induced by an arbitrary linear ordering of formulas to have a more compact representation of the dependency sets. Considering that only subformulas of the conclusion can be in any dependency set, we only need bitstrings of the size of the conclusion of the proof. Figure~\ref{ex:bitstring} shows this final form of tree representing the derivation in figure~\ref{ex:derivacao} and~\ref{ex:tree}, when the linear order $\prec$ is $A\prec B\prec C\prec A\imply B\prec B\imply C\prec A\imply C$. This explanation is an excerpt from \cite{SuperPoly}. 

\begin{figure}[H]
    \begin{tikzpicture}[shorten >=1pt,node distance=1cm,auto]%,on grid
\tikzstyle{state}=[shape=circle,thick,minimum size=1.5cm]
\node[state] (ground) {.};
\node[state, above of=ground] (AC) {$A\imply C$};
\node[state, above of=AC] (C) {$C$};
\node[state, above left of=C] (B) {$B$};
\node[state, above left of=B] (A) {$A$};
\node[state, above right of=B] (AB) {$A\imply B$};
\node[state, above right of=C] (BC) {$B\imply C$};
\path [draw] (ground) edge node[midway] {{\tiny $000110$}} (AC);
\path [draw] (AC) edge node[midway] {{\tiny $100110$}} (C);
\path [draw] (C) edge node[midway,right] {{\tiny $000010$}} (BC);
\path [draw] (C) edge node[midway, left] {{\tiny $100100$}} (B);
\path [draw] (B) edge node[midway, right] {{\tiny $000100$}} (AB);
\path [draw] (B) edge node[midway, left] {{\tiny $100000$}} (A);
    \end{tikzpicture}
\caption{Tree with bitstrings representing the derivation in figure~\ref{ex:derivacao}}\label{ex:bitstring}
\end{figure}

%\tdmargin{Colocar aqui um exemplo de arvore derivada de prova com labels vindo de $\lambda$:{\color{white} FAZER e lembrar que vai ser apontado na def de $F_{\mathcal{S},\Phi}$ Precisa mesmo ???}}   

%\td{Aqui deve vir a mensagem do artigo. Comentar o fato do principio da subformula fazer com que provas grandes sejam altamente redundantes.}

Without loss of generality, we consider the additional hypothesis on the linear bound on height of the proof of \mil tautologies. In \cite{Studia2019}, we show that any tautology in \mil has a Natural Deduction normal proof of height bound by the size of this tautology. However, proof of the tautology does not need to be normal. 
On the other hand, if we consider the complexity class $CoNP$ (see the appendix in ~\cite{Exponential}) we are naturally limited to linearly height-bounded proofs. The proofs, in \mil,  of the non-hamiltonianicity of graphs, are linearly height bounded.

We consider the usual definition of the syntax tree
for \mil-formulas. Given a formula $\phi_1\imply \phi_2$ in \mil, we call $\phi_2$ its right-child and $\phi_1$ its left-child. These formulas label the respective right and left child vertexes. A right-ancestral of a vertex $v$ in a syntax-tree $T_{\alpha}$ of a formula $\alpha$ is any vertex $u$, such that, either $v$ is the right-child of $u$, or, there is a vertex $w$, such that $v$ is the right-child of $w$ and $u$ is right-ancestral of $w$.

The left premise of a $\imply$-Elim rule is called a minor premise, and the right premise is called the major premise. We should note that the conclusion of this rule and its minor premise, are sub-formulas of its major premise. A derivation is a tree-like structure built using $\imply$-Intro and $\imply$-Elim rules. We have some examples depicted in the last section. The derivation conclusion is the root of this tree-like structure, and the leaves are what we call top-formulas. A proof is a derivation that has every top-formula discharged by a $\imply$-Intro application in it. The top-formulas are also called assumptions. An assumption that it is not discharged by any $\imply$-Intro rule in a derivation is called an open assumption. If $\Pi$ is a derivation with conclusion $\alpha$ and $\delta_1,\ldots,\delta_n$ as all of its open assumptions then we say that $\Pi$ is a derivation of $\alpha$ from $\delta_1,\ldots,\delta_n$.

\begin{definition}[Branch]\label{def:Branch} A branch in a derivation or proof $\Pi$ is any sequence $\beta_1,\ldots,\beta_k$ of formula occurrences in $\Pi$, such that:
\begin{itemize}
\item $\delta_1$ is a top-formula, and;
\item For every $i=1,k-1$,  either $\beta_i$ is a $\imply$-Elim major premise of $\beta_{i+1}$ or $\beta_i$ is a $\imply$-Intro premise of $\beta_{i+1}$, and;
\item $\delta_k$ either is the conclusion of the derivation or the minor premise of a $\imply$-Elim.
\end{itemize}
\end{definition}

A normal derivation/proof in \mil is any derivation that does not have any formula occurrence simultaneously a major premise of a $\imply$-Elim and the conclusion of a $\imply$-Intro. A formula occurrence that is a conclusion of a $\imply$-Intro and a major premise of $\imply$-Elim is called a maximal formula. 
In \cite{Prawitz}, Dag Prawitz proves the following theorem for the Natural Deduction system for the full\footnote{The full propositional fragment is $\{\lor, \land, \imply, \neg, \bot\}$} propositional fragment of minimal logic.      

\begin{theorem}[Normalization]
  Let $\Pi$ be a derivation of $\alpha$ from $\Delta=\{\delta_1,\ldots, \delta_{n}\}$. There is a normal proof $\Pi^{\prime}$ of $\alpha$ from $\Delta^{\prime}\subseteq\Delta$.
\end{theorem}

In any normal derivation/proof, a branch's format provides worth information on why huge proofs are redundant, as we will see in the next sections. Since no formula occurrence can be a major premise of $\imply$-Elim and conclusion of a $\imply$-Intro rule in a branch we have that the conclusion of a $\imply$-Intro can only be the minor premise of a $\imply$-Elim or it is not a premise of any rule application at all in the same branch. In this last case, it is the derivation's conclusion or the minor premise of a $\imply$-Elim rule. In any case, it is the last formula in the branch. Thus, for any branch, every conclusion of a $\imply$-Intro has to be a premise of a $\imply$-Intro. Hence, any branch in a normal derivation splits into two parts (possibly empty). The E-part starts it with the top-formula, and, every formula occurrence in it is the major premise of a $\imply$-Elim. We may have then a formula occurrence that is the conclusion of a $\imply$-Elim and premise of a $\imply$-Intro rule that is called minimal formula of the branch. The minimal formula starts the I-part of the branch, where every formula is the premise of a $\imply$-Intro, excepted the last formula of the branch. From the branches' format, we can conclude that the sub-formula principle holds for normal proofs in Natural Deduction for \mil, in fact, for many extensions. A branch in $\Pi$ is said to be a principal branch if its last formula is the conclusion of $\Pi$. A secondary branch is a branch that is not principal. The primary branch is called a 0-branch.  Branches, where the last formula is the minor premise of a rule in the E-part of a $n$-branch, is a $n+1$-branch.

\begin{corollary}[Sub-formula principle]\label{coro:SubForProperty}
  Let $\Pi$ be a normal derivation of $\alpha$ from $\Delta=\{\delta_1,\ldots,\delta_m\}$. It is the case that for every formula occurrence $\beta$ in $\Pi$, $\beta$ is a sub-formula of either $\alpha$ or of some of $\delta_i$.
\end{corollary}

To facilitate the presentation, we only handle normal proofs in expanded form.

\begin{definition} A normal proof/derivation is in expanded form, if and only if, all of its minimal formulas are atomic.
\end{definition}

Without loss of generality, we can consider that formula in $\mil$ is a tautology if and only if there is a normal proof in expanded form that proves it. Of course, if it is a tautology, it has proof, and so it has normal proof by normalization. We use the following fact to obtain the expanded form from normal proof. In \cite{SuperPoly} we prove that all tautologies have normal proofs in expanded form. See  the first appendix of \cite{SuperPoly}

%{\color{red}Fazer desenho do lemma abaixo~\ref{lemma:Subformula}}

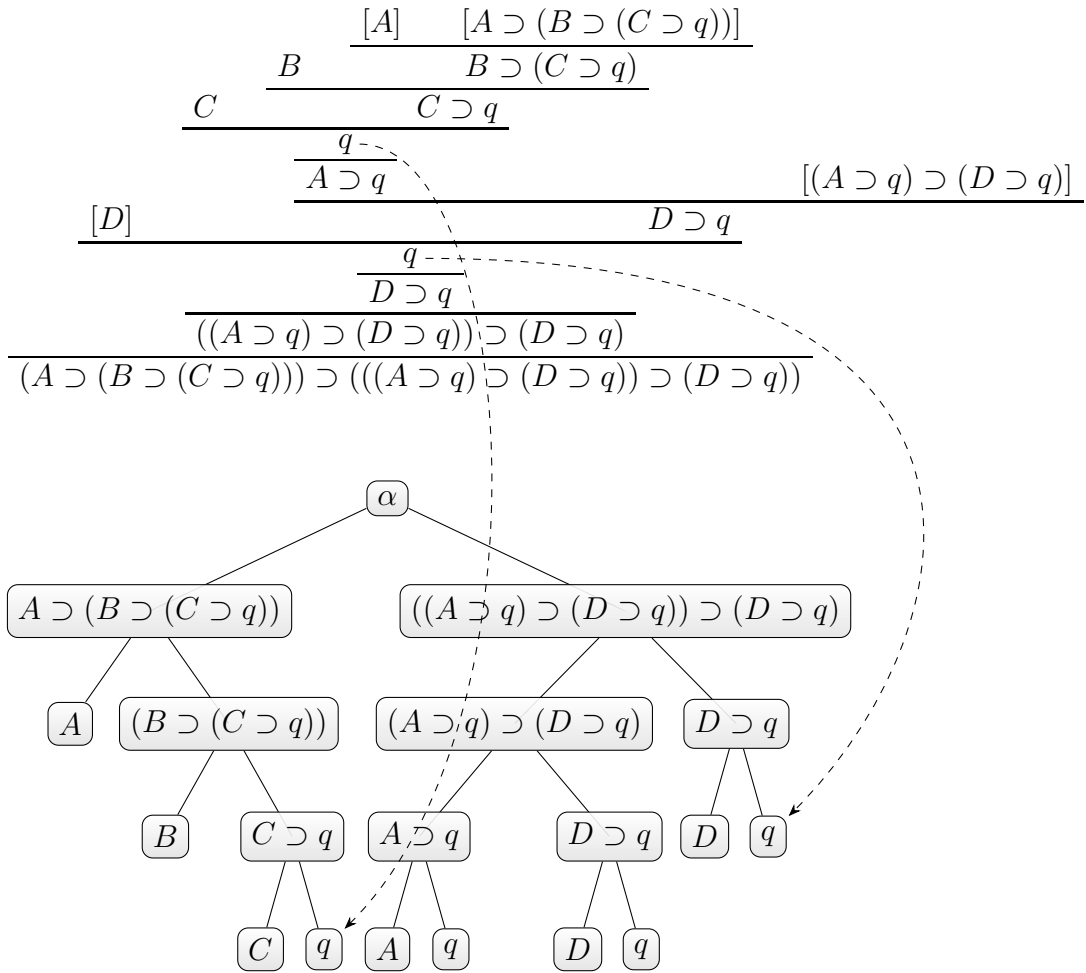
\begin{figure}
  \centering
 \centering
  \begin{tikzpicture}
    \node[anchor=south west] (proof) at (0,0) {%
        \AxiomC{$[D]$}
        \AxiomC{$C$}
        \AxiomC{$B$}
        \AxiomC{$[A]$}
        \AxiomC{$[A\imply (B\imply (C\imply q))]$}
        \BinaryInfC{$B\imply (C\imply q)$}
        \BinaryInfC{$C\imply q$}
        \BinaryInfC{$q$}
        \UnaryInfC{$A\imply q$}
        \AxiomC{$[(A\imply q)\imply (D\imply q)]$}
        \BinaryInfC{$D\imply q$}
        \BinaryInfC{$q$}
        \UnaryInfC{$D\imply q$}
        \UnaryInfC{$((A\imply q)\imply (D\imply q))\imply (D\imply q)$}
        \UnaryInfC{$(A\imply (B\imply (C\imply q)))\imply (((A\imply q)\imply (D\imply q))\imply (D\imply q))$}
    \DisplayProof};
  \begin{scope}[x={(proof.south east)},y={(proof.north west)}]
    \node (ANCHOR) at (.38,.36) {};
    \node (ANCHOR1) at (.32,.64) {};
  \end{scope}
  \begin{scope}[sibling distance=15em,
      every node/.style={shape=rectangle, rounded corners,
        draw, align=center,top color=white, bottom color=gray!20}]
  \node[below=of proof,xshift=-5em] {$\alpha$}
  child { [sibling distance=5em] node {$A\imply (B\imply (C\imply q))$}
    child { node {$A$} }
    child { [sibling distance=4em] node {$(B\imply (C\imply q))$}
      child {  node {$B$} }
      child { [sibling distance=2em] node {$C\imply q$}
        child { node {$C$} }
        child { node (b) {$q$} }
            }
          }
        }
    child { [sibling distance=7em] node {$((A\imply q)\imply (D\imply q))\imply (D\imply q)$}
      child { [sibling distance=6em] node {$(A\imply q)\imply (D\imply q)$}
        child { [sibling distance=2em] node {$A\imply q$}
          child { node {$A$} }
          child { node {$q$} }
              }
        child { [sibling distance=2em] node {$D\imply q$}
          child {  node {$D$} }
          child {  node (c) {$q$} }
              }
            }
      child { [sibling distance=2em]  node {$D\imply q$}
          child {  node {$D$} }
          child {  node (c) {$q$} }
    } };
    \path (0,0) rectangle (1,-2);
    \draw[->,shorten >=1pt,>=Stealth,dashed] (ANCHOR)
    to [out=0, in=45,out looseness=.5, out min distance=20em] (c);
    \draw[->,shorten >=1pt,>=Stealth,dashed] (ANCHOR1)
    to [out=0, in=45,out looseness=.5, out min distance=-20em] (b);
  \end{scope}
  \end{tikzpicture}
  \caption{A mapped N.D. proof}
  \label{figure:MappedProofs}
\end{figure}

In \cite{SuperPoly}, we observed correspondence between each minimal formula of a branch, in a normal and expanded proof $\Pi$, employing a one-to-one correspondence to the respective top-formula occurrence of its branch. Figure~\ref{figure:MappedProofs} illustrates the mapping from the proof into the syntax-tree of the proved formula according to this correspondence. Note that the two positions of the atomic formula $q$ in the syntax tree uniquely indicates the top-formula in the E-part of the Natural Deduction proof/derivation to which it belongs. We can consider that the two $q$'s are in fact, different. The top-formula of each $q$ is the biggest in the inverse path (upwards) following the reverse of the right child edge.  
 Definition~\ref{def:E-mapped-ND}, in the sequel, has the purpose of setting this correspondence in the representation of proofs formally.  With the sake of a more precise presentation, we provide below the definition of a formula's syntax tree.

\begin{definition}[Syntax tree of a formula]
  Let $\alpha$ be a \mil-formula. The syntax tree of $\alpha$ is the triple $\langle V,E_{left},E_{right},L\rangle$ where $V$ is a set, of vertexes, $E_{s}\subseteq V\times V$, $s=left,right$, the corresponding left and right edges, such that $\langle V,E_{left},E_{right}\rangle$ is an ordered full binary tree, and, $L$ is a bijective function from $V$ onto the subformulas of $\alpha$, such that:
  \begin{itemize}
  \item $L(r)=\alpha$, where $r\in V$ is the root of the tree $\langle V,E_{left},E_{right}\rangle$, and;
  \item For every formula $\varphi_1\imply\varphi_2\in Sub(\alpha)$, if $L(v)=\varphi_1\imply\varphi_2$,  $\langle v,v_1\rangle\in E_{left}$ and $\langle v,v_2\rangle\in E_{right}$ then $L(v_1)=\varphi_1$ and $L(v_2)=\varphi_2$.
  \end{itemize}
  \end{definition}

\begin{definition}[Partially mapped ND-proofs]\label{def:Mapped-ND}
  Let $\alpha$ be a \mil-formula and $T_{\alpha}=\langle V,E_{left},E_{right},L\rangle$ its syntax tree. Let $\Pi$ be a \mil-ND normal derivation of $\alpha$. A partially mapped ND-proof of $alpha$ is a structure $\langle \Pi,T_{\alpha},l\rangle$, where $l$ is a partial function from the formula occurrences in $\Pi$ to $V$, such that, the following conditions hold.
  \begin{itemize}
  \item If $\gamma$ is the minimal formula of a branch $\Branch{b}$ in $\Pi$ then if $\ell(\gamma)$ is defined then $L(\ell(\gamma))=\gamma$;
    \item If $\gamma$ is the minimal formula of a branch $\Branch{b}=\langle b_0,\ldots,b_j=\gamma,\ldots,b_k\rangle$ and $\ell(\gamma)$ is defined then either $\ell(b_{j-1})$ or $\ell(b_{j+1})$ are defined, and;
    \item If $\varphi_2$ is the conclusion of a $\imply$-Elim rule in $\Pi$, that has premises $\varphi_1\imply\varphi_2$ and $\varphi_1$,  and $\ell(\varphi_2)=v_2$ then there are $v$ and $v_1$, such that $\langle v, v_2\rangle\in E_{right}$, $\langle v,v_1\rangle\in E_{left}$, $\ell(\varphi_1)=v_1$ and $\ell(\varphi_1\imply\varphi_2)=v$;
          \item If $\varphi_1\imply\varphi_2$ is the conclusion of a $\imply$-Intro rule in $\Pi$, that has premise $\varphi_2$ and $\ell(\varphi_1\imply\varphi_2)=v$ then there is $v^{\prime}\in V$, $\langle v,v^{\prime}\rangle\in E_{right}$ and $\ell(v^{\prime})=\varphi_2$. 
      \end{itemize}
\end{definition}

\begin{definition}[E-mapped Natural Deduction Normal Expanded proofs]\label{def:E-mapped-ND}
  Let $\alpha$ be a \mil-formula, $T_{\alpha}=\langle V,E_{left},E_{right},L\rangle$ be the syntax tree of $\alpha$ and $\Pi$ a normal and expanded proof of $\alpha$. The triple $\langle \Pi,T_{\alpha},l\rangle$ is an E-mapped Natural Deduction proof, if and only if, $l$ is defined on all formula occurrences that take part in the E-parts of branches in $\Pi$, including the minimal formulas. Moreover the following condition must hold:
  \begin{itemize}
  \item For every branch $\Branch{b}$, if $q$ is the minimal formula of $\Branch{b}$, $\ell(q)=v\in V$ and $\beta$ is the top-formula (occurrence) of $\Branch{b}$ then $\ell(\beta)=u$, where $u$ is the right-ancestral of $v$ that is left-child of some $w\in V$.
  \end{itemize}
  \end{definition}

In \cite{SuperPoly} we show that the above definition of E-mapped Natural Deduction Normal Expanded proof, {\bf EmND} for short,   is well-defined. Moreover, we have the following proposition. We consider a branch as a sequence of formula occurrences numbered from top-formula down to the branch's last formula. The proof of this proposition is in \cite{SuperPoly}

\begin{proposition}\label{prop:one2one}
  Let $\langle \Pi,T_{\alpha},l\rangle$ be a {\bf EmND} of $\alpha$. We have that to each branch $\Branch{b}=\langle \beta_0,\ldots,\beta_k,$ in $\Pi$ that has the  minimal formula occurrence $q=\beta_{j}$, such that $\ell(\beta_{j})=u\in V$,  there exists one and only one path $p=\langle u_0,\ldots,u_j\rangle$ in $T_{\alpha}$, with $u_j=u$, and $u_0$ such that, $\ell(\beta_{i})=u_i$, $i=0,\ldots,j$. 
\end{proposition}

We point out that in proposition~\ref{prop:one2one}, above,  $\langle \beta_0,\ldots,\beta_j=q\rangle$ is the E-part of $\Branch{b}$.  This proposition~\ref{prop:one2one} states that any given E-part $\langle \beta_0,\ldots,\beta_j\rangle$ of a branch in an {\bf EmND} is an instance of at most one path $p=\langle u_0,\ldots,u_j\rangle$ in $T_{\alpha}$, such that $L(u_i)=\beta_i$, $i=0,\ldots,j$. Moreover, this path $p$ is as stated in the definition of {\bf EmND} in the its only item.  Given a {\bf EmND} $\Pi$,
for  each  E-part,  in an {\bf EmND},  exists a path of the form stated in definition of {\bf EmND}, in the syntax tree of the conclusion of the{\bf EmND}. The number of such paths in the syntax tree is upper-bounded by its size, then the number of different E-parts types in any {\bf EmND} is at most of the size of the conclusion of this {\bf EmND}. We have the following lemma:

\begin{lemma}[Linear upper-bound on types of E-parts]\label{lemma:Upper-boundsE-parts}
  Let $\Pi$ be an {\bf EmND} proving the \mil-formula $\alpha$. The number of different types of E-parts occurring in this {\bf EmND} is at most the size of the $T_{\alpha}$.
\end{lemma}

We remark, as also observed in \cite{SuperPoly},  that we can label the nodes of a Natural Deduction proof-tree with the nodes (not the labels) of the syntax tree of the conclusion of the proof-tree. In doing that we will have the same effect on counting different types of E-parts that is stated by  lemma~\ref{lemma:Upper-boundsE-parts}.

%\section{Remembering how to Count repeated patterns in polynomially lower-bounded proofs}\label{sec:Lower-bounds}

\subsection{Redundancy in huge $\mathcal{M}_{\imply}$ mapped derivations}\label{subsec:Main}

Due to the linear speedup theorem, see \cite{LinSpeedUp-Space} page 63-64, Theorem 3.10, we can consider, w.l.o.g a linear height bounded proof of $\alpha$ is a proof which height is upper-bounded by the length of $\alpha$. In fact, in this article, because of counting details, we consider that the upper-bounded is the size of the syntax tree $\size{T_{\alpha}}$. Since $\size{T_{\alpha}}=\size{\alpha}$, the definition is equivalent. From \cite{SuperPoly}, we have the following lemmas. In \cite{SuperPoly} the reader can fund both proofs.   

  \begin{lemma}[Spreading Branchs Repetitions]\label{Zero} Let $\langle \Pi,T_{\alpha},l\rangle$ be a linearly height bounded {\bf EmND} proof of $\alpha$, $0<p\in \mathbb{N}$ and $m=\size{\alpha}$. If there is a branch $\Branch{b}$ that has more than $m^p$ instances occurring in $\Pi$ then there is a level $\mu$, such that, at least $m^{p-1}$ instances of $\Branch{b}$ have the minimal formula $q_{\Branch{b}}$ of $\Branch{b}$ occurring in level $\mu$.
    \end{lemma}

  \begin{lemma}[Branchs and sub-derivations]\label{Um} Let $\Pi$ be a proof of $\alpha$, and $\Branch{b}$ a branch in $\Pi$ under the same conditions of the lemma~\ref{Zero} above. Then there is a (sub) derivation $\Pi_{\Branch{b}}$ of $\Pi$, such that, $\Pi_{\Branch{b}}$ has at least $m^{p-1}$ instances occurring in $\Pi$.
  \end{lemma}

%  {\color{red} Aqui nesta seÃ§Ã£o considerar as alterÃ§Ãµes nos upper-bounds e a nÃ£o normalidade da prova}

\begin{definition}[Linearly height-bounded EmND proofs]\label{def:SLambda}
  Let $\Lambda$ be the set of mapped linearly height-bounded ND \mil proofs. We use the notation $c(\Pi)$ to denote the formula that is the conclusion of $\Pi$. Note that we can consider $\Lambda$ as a predicate $\Lambda(x)$ that is true if and only if $x$ is assigned to a mapped linearly height-bounded ND proof. 
\[
S_{\Lambda}=\{\mbox{$\Pi\in \Lambda$}:\mbox{$\forall p\in\mathbb{N}, p>0, \exists n_0,\forall n>n_0$, $\size{T_{c(\Pi)}}=n$ and $\size{\Pi}>n^p$)}\}
\]
\end{definition}
%{\color{purple} Motivar muito bem essa definiÃ§Ã£o e a relaÃ§Ã£o com NP=PSPACE ou NP=CoNP}

As explained in \cite{SuperPoly},  $S_{\Lambda}$ contains all huge or hard linearly height upper-bounded proofs in \mil. Of particular interest is the following set. Let $Taut_{\mil}$ be the set of all ND mapped proofs of \mil tautologies. The following set:

\begin{definition} Let $\Theta$ be the following set:
  \[
  \Theta_{\mil} = \{\mbox{$\Pi\in Taut_{\mil}$}:\mbox{$\forall p\in\mathbb{N}, p>0, \exists n_0,\forall n>n_0$, $\size{T_{c(\Pi)}}=n$ and $\size{\Pi}>n^p$}\}
  \]
\end{definition}

$\Theta_{\mil}$ is the set of super-polynomially sized \mil ND mapped proofs.

In \cite{SuperPoly}, we show that every $\Pi\in S_{\Lambda}$ is redundant. This means that there is at least one sub-proof $\Pi_{s}$  of $\Pi$ that repeats as many times as it is the size of $\Pi$. We have the following theorem~\ref{main}, proved in \cite{SuperPoly}. We emphasise that the proofs in $S_{\Lambda}$ are linearly heigh-bounded.

\begin{theorem}\label{main} For all $p\in\mathbb{N}$, $p>3$, and for all $\Pi\in S_{\Lambda}$, such that, $\size{\mathcal{T}(c(\Pi))}=m$ and $\size{\Pi}>m^p$,  then there is a sub-derivation $\Pi_{s}$ of $\Pi$ and a level $\mu$ in $\Pi$, such that, $\Pi_{s}$ has at least $m^{p-3}$ instances occurring in the level $\mu$ in $\Pi$.
\end{theorem}

%\begin{corollary}\label{roughly1} In every family of super-polynomial proofs, all of then are super polynomial redundant.
    %\end{corollary}

  From theorem~\ref{main} we can roughly state the corollary~\ref{roughly}.

\begin{corollary}\label{roughly} All, but finitely many, proofs belonging to an arbitrary family of super-polynomial and linearly height upper-bounded proofs, are super polynomial redundant.
\end{corollary}

\section{Removing redundancies from huge \mil linearly height bounded proofs}\label{sec:RemovingRedundancies}

Corollary~\ref{roughly} says that any proof in an unbounded set of super-polynomial proofs, linearly bounded on the height, is almost as redundant as it is huge.  We can show that there are level $\mu$ and a derivation  $\Pi$ that occurs as many times in $\mu$ as it is the size of the proof.  This section shows a polynomial sized certificate of validity for any huge tautology that belongs to this set of super-polynomial proofs. Our argumentation to prove this is to remove all redundancies in the original derivation, preserving logical consequence.

Given a non-empty finite set $S$, $card(S)=n$, and a total order $\mathcal{O}_{S}=\{s_1,\ldots,s_n\}$ on $S$, the set $B(\mathcal{O}_{S})$ is $\{b_{1}\ldots b_{n}:\mbox{$b_{i}=0$ or $b_{i}=1$, $i=1\ldots n$}\}$. There is a bijection $F$ from $B(\mathcal{O}_{S})$ onto the powerset of $S$ given by $F(b_1\ldots b_{n})=\{s_i:b_{i}=1\}$. $B(\mathcal{O}_S)$ is also called the set of bitstrings over $\mathcal{O}_S$. 

Given a \mil formula $\alpha$, $sub(\alpha)$ is the set of all sub-formulas of $\alpha$. 

\begin{definition}[r-DagProof]\label{def:Prer-DagProof}
 A {\bf pre} r-DagProof for a \mil formula $\alpha$  is a structure $\mathcal{C}$ = $\langle V , E_{d}, E_{A}, r, \ell, L,\rho,\delta,\mathcal{O}_{\alpha}\rangle$
  \begin{enumerate}
  \item $V$ is a non-empty set of nodes;
  \item $E_{d}\subseteq V\times V$, deduction edges;
  \item $E_{A}\subseteq V\times V$, ancestrality edges;
    \item $\mathcal{O}_{\alpha}$ is a total order on $sub(\alpha)$;
  \item $r\in V$ is the root of the $\mathcal{C}$;
  \item $\ell:V\rightarrow sub(\alpha)$, for $v\in V$, $\ell(v)$ is the (formula) label of $v$, $sub(\alpha)$ is the set of all sub-formulas of $\alpha$;
  \item $L:E_{d}\rightarrow \mathcal{B}(\mathcal{O})$ is a function, such that, for every $\langle u,v\rangle\in E_{d}$, $L(\langle u,v\rangle)\in B(\mathcal{O}_S)$;
  \item $\delta:E_{A}\rightarrow \mathbb{N}$, a total function;
    \item $\rho:E_{d}\rightarrow \mathbb{N}$, a partial function; 
  \end{enumerate}
  Subjected to the following conditions:
  %  \begin{enumerate}[label=(\alph*)]
  \begin{description} 
   \item [(Global)] $\langle V, E_{d}, r\rangle$ is a connected DAG with unique root $r$ and for each node $v\in V$, if $\langle r=u_1,\ldots, u_m=v\rangle$ and $\langle r=v_1,\ldots, v_n=v\rangle$ are two inverse paths from $r$ to $v$ then $m=n$, i.e., for any $v$, any longest path from $r$ to $v$ has the same length of a shortest path from $r$ to $v$. Moreover, for each $v\in V$, if $v\neq r$ then there is a path from $v$ to $r$ and for every $u$, $\langle u,u\rangle\not\in E_d$. 
    \item[(${\bf E_d-(l/L)_1}$ consistency)] For every $\langle u,v\rangle,\langle v,w\rangle\in E_{d}$,  if there is no $\langle u^{\prime},v\rangle\in E_{d}$, such that, $u^{\prime}\neq u$ and $\rho(\langle v,w\rangle)=\uparrow$,  then $\ell(u)=\varphi_1$, $\ell(v)=\varphi_2\imply\varphi_1$, $L(\langle v,w\rangle)=L(\langle u,v\rangle)-\vec{b}_{\varphi_2}$;   
    \item [(${\bf E_d-(l/L)_2}$ consistency)] For every $\langle u_1,v\rangle,\langle u_2,v\rangle,\langle v,w\rangle\in E_{d}$, there is no $\langle u^{\prime},v\rangle\in E_{d}$, such that, $u^{\prime}\neq u_i$, $i=1,2$,  and $\rho(\langle v,w\rangle)=\uparrow$,  then $\ell(u_1)=\varphi_1$, $\ell(u_2)=\varphi_1\imply\varphi_2$,  $\ell(v)=\varphi_2$, $L(\langle v,w\rangle)=L(\langle u_1,v\rangle)oplus L(\langle u_2,v\rangle$;
    \item [(${\bf E_{A}}$-target consistency)] If $\langle u,v\rangle\in E_{A}$ then either there is no $\langle w,v\rangle\in E_{d}$ or there is $\langle v,w\rangle\in E_{A}$, $w\neq u$;
    \item [(${\bf E_{A}}$-source consistency)] If $\langle u,v\rangle\in E_{A}$  then there is $\langle w,u\rangle\in E_{d}$, $w\neq v$ and $\delta(\langle u,v\rangle)=\rho(\langle w,u\rangle)$
    \item [(${\bf E_{A}}$-ancestrality irreflexivity)] For every $u$, $\langle u,u\rangle\not\in E_{A}$. 
      
    \end{description}
\end{definition}
  
  %  \end{enumerate}

A {\bf pre} r-DagProof, satisfying definition~\ref{def:Prer-DagProof} above, is a r-DagProof, or  a certificate for $\ell(r)$, if it is sound, or equivalently, if  algorithm~\ref{algo:Check-rDagProof} answers ``Correct'' when executes on it. In section~\ref{sec:PolytimeCompression} we define soundness of a r-DagProof and we define the algorithm~\ref{algo:Check-rDagProof}. In the meanwhile we deal with {\bf pre} r-DagProofs. Specifically, in this section, this difference between {\bf pre} r-DagProofs and r-DagProofs is not relevant.

We provide a bit of terminology in this part of the article. Regarding theorem~\ref{main}, we denote by {\bf matrix} the sub-derivation of the ambient derivation, i.e. the huge one, that has many instances repeated. Definition~\ref{def:Initials}  define the set of nodes in an r-DagProof that is formed by top-formulas or their representatives after collapsing. We need the definition of top-formulas and representative top-formulas.

  \begin{definition}[Top-formula of a r-DagProof]\label{def:Top-formula}
    A node $v\in V$, of a r-DagProof $\mathcal{C}=\langle V , E_{d}, E_{A}, r, l, L, P,\mathcal{O}_{\alpha},E_{\delta}\rangle$,  is a top-formula, if and only if, there is no $w\in V$, such that $\langle w,v\rangle\in E_{d}$ or $\langle w,v\rangle\in E_{A}$.
    \end{definition}

    \begin{definition}[Representative top-formula of a r-DagProof]\label{def:Rep-Top-formula}
      Given a r-DagProof  $\mathcal{C}$=$\langle V , E_{d}, E_{A}, r, l, L, \rho,\delta,\mathcal{O}_{\alpha}\rangle$, $v\in V$ is a representative top-formula, if and only if, there is no $w\in V$, such that $\langle w,v\rangle\in E_{A}$. Moreover, there is a sequence $v=w_1,\ldots,w_n$, $w_i\in V$, $i=1,n$, such that, for every $i=1,n-1$, $\langle w_{i},w_{i+1}\rangle\in E_{A}$, and there is no $w\in V$, such that, $\langle w,w_n\rangle\in E_{d}$. 
    
    \end{definition}

  \begin{definition}\label{def:Initials}
Given a r-DagProof $\mathcal{C}=\langle V , E_{d}, E_{A}, r, \ell, L,\rho,\delta,\mathcal{O}_{\alpha}\rangle$, its set of initials, $I(\mathcal{C})$ is the set of representative top-formulas together with its top-formulas.

    \end{definition}

  The following mapping defines the detachment of an instance of the matrix $\mathcal{C}$, defined by its root $k$,
  from $\mathcal{D}$ and link its position to the matrix $\mathcal{C}$ accordingly. In the following definitions, if $\mathcal{C}$ is a r-DagProof   $\langle V , E_{d}, E_{A}, r, l, L, \rho,\delta,\mathcal{O}_{\alpha}\rangle$, we use the notations $E_{A}(\mathcal{C})$, $E_{d}(\mathcal{C})$, $V(\mathcal{C})$, $r(\mathcal{C})$, $\ell(\mathcal{C})$, etc, to denote $E_{A}$, $E_{d}$, $V$, $r$, $l$, etc, respectively. Given a set of nodes $V^{\prime}\subseteq V$, such that $r\not\in V^{\prime}$, the restriction of $\mathcal{C}$ to the set of nodes $V^{\prime}$ is $\mathcal{C}\mid_{V^{\prime}}$ and it is defined below:
  \[
  \langle V^{\prime}, E_d\mid_{V^{\prime}}, E_{A}\mid_{V^{\prime}}, r, l\mid_{V^{\prime})}, L\mid_{(E_{d}\mid_{V^{\prime}})}, \rho\mid_{(E_{d}\mid_{V^{\prime}})}, \delta\mid_{(E_{A}\mid_{V^{\prime}})}, \mathcal{O}_{\alpha}\rangle 
  \]
  ,where if $E\subseteq V\times V$ is a set of edges and $V^{\prime}\subseteq V$ then $E\!\!\mid_{V^{\prime}}=\{\langle v,u\rangle:\mbox{$v\in V^{\prime}$ and $u\in V^{\prime}$}\}$.

  \begin{definition}[Difference of r-DagProofs] Let $\mathcal{C}$ and $\mathcal{D}$ be two r-DagProofs, the graph difference $\mathcal{D}-\mathcal{C}$ is $\mathcal{D}\!\!\mid_{(V(\mathcal{D})-V(\mathcal{C}))}$.

    \end{definition}

  The difference of r-DagProofs is not a r-DagProof itself. In the definition below we use the notation $\mathcal{D}\!\!\uparrow\!\!\!(k)$, where $\mathcal{D}$ is a r-DagProof and $k\in V(\mathcal{D})$, to denote the biggest sub r-DagProof of $\mathcal{D}$ that has $k$ as root. In the particular case that $\mathcal{D}$ is a tree, $\mathcal{D}\!\!\uparrow\!\!\!(k)$ is the sub-tree of $\mathcal{D}$ that has root $k$.

  \begin{definition}[Detach and link sub r-DagProofs function]\label{def:Detach}
    Let $\mathcal{D}$ be a r-DagProof of $\alpha$ and $k$ a node of $\mathcal{D}$ that it is the root of an instance of
    the r-DagProof $\mathcal{C}$ that is its matrix, as given by theorem~\ref{main}. Let $i\in\mathbb{N}$, a label, and $\mathcal{C}^{\prime}=\mathcal{D}\!\!\uparrow\!\!\!(k)$.
    Define $\mathcal{D}^{\prime}=(\mathcal{D}-\mathcal{C}^{\prime})\cup\mathcal{C}$. % {\color{red} EXPLICAR ESSA OPERAÃÃO}
    We define $DetachLink(\mathcal{D},k,\mathcal{C},i)$ as following:
    {\small
      \[
      \begin{array}{l}
        \langle V(\mathcal{D}^{\prime}), \\
        E_d(\mathcal{D}^{\prime})\cup \{\langle r(\mathcal{C}),v\rangle:\langle k,v\rangle\in E_{d}(\mathcal{D})\}, \\
        E_{A}(\mathcal{D}^{\prime})\cup \{\langle v,w\rangle:\mbox{$\langle k,v\rangle\in E_{d}(\mathcal{D})$ and $w\in I(\mathcal{C})$}\}, \\
        r(\mathcal{D}^{\prime}),\\
        \ell(\mathcal{D}^{\prime}),\\
        L(\mathcal{D}^{\prime})\cup\{\langle r(\mathcal{C}),v\rangle\mapsto L(\langle k,v\rangle):\langle k,v\rangle\in E_{d}(\mathcal{D})\},\\
        \rho(V(\mathcal{D}^{\prime}))\cup\{\langle v,w\rangle\mapsto i:\mbox{$\langle k,v\rangle\in E_{d}(\mathcal{D})$ and $w\in I(\mathcal{C})$}\},\\
        \delta(V(\mathcal{D}^{\prime}))\cup \{\langle r(\mathcal{C}),v\rangle\mapsto i:\langle k,v\rangle\in E_{d}(\mathcal{D})\},\\
        \mathcal{O}_{\alpha}\rangle
      \end{array}
    \]
}    

     \end{definition}
  Figures~\ref{fig:DetachLinkI} and~\ref{fig:DetachLinkII} illustrate what happens when we apply three DetachLink operations in an ambient r-DagProof. When we repeat this operation to every instance of a matrix $\mathcal{C}$ occurring in a fixed level $\mu$, we say that we performed a collapse of the instances of $\mathcal{C}$ in $\mathcal{D}$ in level $\mu$. This operation is described by algorithm~\ref{algo:Collapse} and is denoted by $Collapse(\mathcal{C},\mathcal{D},\mu)$.

  Consider a {\bf EmND} proof that has more than one matrix $\mathcal{C}$ with instances occurring super-polynomially many times in it. Thus, Lemma~\ref{lemma:ListForCollapse} shows that we can use theorem~\ref{main} to obtain the list of all matrices having a super-polynomially many instances occurring in a fixed level $\mu$ of this {\bf EmND} proof. We remember that a matrix derivation/proof, in our terminology, is nothing but a sub-derivation/sub-proof that has at least one instance in other proof.
  This lemma obtains a set of Matrices having all instances occurring in the lowest level in a specific EmND proof. Moreover, no matrix instance is a proper sub-derivation of any other matrix instance in the set. We call this an independent set of matrices from $\Pi$ to this set.

  \begin{definition}[Independent set of matrices in a proof/derivation]\label{def:IndependentSetMatrices} Let $\Pi$ be an EmND proof/derivation. A set of matrices
    \[
    S=\{\Pi_{\nu}:\mbox{$\Pi_{\nu}$ is a matrix in $\Pi$ which only has instance in a level $\nu$}\}
    \]
    is an independent set of matrices, if and only if,  there is no $\nu_1$ and $\nu_2$ levels in $\Pi$, $\nu_1\neq\nu_2$, and instances $\Pi_{1}$ and $\Pi_2$ of $\Pi_{\nu_1}$ and $\Pi_{\nu_2}$, respectively, $\Pi_{\nu_i}\in S$, $1=1,2$, such that, $\Pi_1$ is a sub-derivation of $\Pi_2$, or $\Pi_2$ is a sub-derivation of $\Pi_1$. 
  \end{definition}

  Given a set $S$ of independent matrices of a proof/derivation,  we can prove that if $\Pi_{\nu}\in S$ then no instance of $\Pi_{\nu}$ is sub-derivation of $\Pi_{\mu}\in S$, unless $\Pi_{\mu}=\Pi_{\nu}$. Thus, if $\Pi_{\nu}\in S$ then there is no level $\mu<\nu$ with instances that are super-derivations of instances of $\Pi_{\nu}$. In a certain sense, $\nu$ is a local lowest level. 

      \begin{lemma}[List of super-polynomially repeated matrices]\label{lemma:ListForCollapse}
    For all $p\in\mathbb{N}$, $p>3$, and for all $\Pi\in S_{\Lambda}$\footnote{The set $S_{\Lambda}$ is defined in defintion~\ref{def:SLambda}}, such that, $\size{\mathcal{T}(c(\Pi))}=m$ and $\size{\Pi}>m^p$,  then there is a set $M$ of independent matrices for sub-derivations instances of $\Pi$, such that, for every $\Pi_{s}\in M$,  $\Pi_{s}$ has at least $m^{p-3}$ instances occurring in some level $\xi$ in $\Pi$.
  \end{lemma}

  \begin{proof} of lemma~\ref{lemma:ListForCollapse}. Using the conditions on lemma, Theorem~\ref{main} provides at least one level $\mu$ and a matrix $\Pi_{s}$ that has at least $m^{p-3}$ instances occurring in $\mu$. Thus, the sets  
    \[
    S_{\nu}=\{\Pi_{s}:\mbox{$\Pi_{s}$ is a matrix having at least $m^{p-3}$ instances in $\nu$ in $\Pi$}\}
    \]
    where $\nu$ is a level in $\Pi$, form a family. The family $(S_{\nu})_{\nu\in Lev(\Pi)}$  has at least the non-empty set $S_{\mu}$.
Moreover, if $\nu_1<\nu_2$ and $S_{\nu_1}\neq\emptyset$ then $S_{\nu_2}$ contains all the subtrees of the trees in $S_{\nu_1}$ that occurs in level $\nu_2-\nu_1$ in the elements (trees) in $S_{\nu_1}$. If $T$ is a subtree of $T^{\prime}$ then we can say that $T^{\prime}$ is a super-tree of $T$. The same applies to sub-derivations and super-derivations.
    We define the set $L$ of levels $\xi$, such that $S_{\xi}$ has no instance in $\Pi$ in level $\xi$ that is a super-derivation of an instance of $\Pi_{s}\in S_{\mu}$, $\mu\neq\xi$. The set $\{\Pi_{\xi}:\mbox{$\xi\in L$}\}$ is an independent set of matrices. It is the biggest one, indeed.  
%    The inclusion of sub-derivation is an ordering in the levels that ensures that the family's member's ordering is well-defined. 

  \end{proof}

  Lemma~\ref{lemma:ListForCollapse} is used to provide the initial list of instances to collapse in the algorithm~\ref{algo:CompressEmNDProof}. This list can be alternatively defined as, given $p>3$, the lowest sub-derivations of a proof $\Pi$, of size bigger than $\size{\Pi}^p$, that occurs at least $\size{\Pi}^{p-3}$ times in their respective lowest level $\mu$. Definition~\ref{def:LRI} introduces the notation $LRI(\Pi)$ to denote this set.

    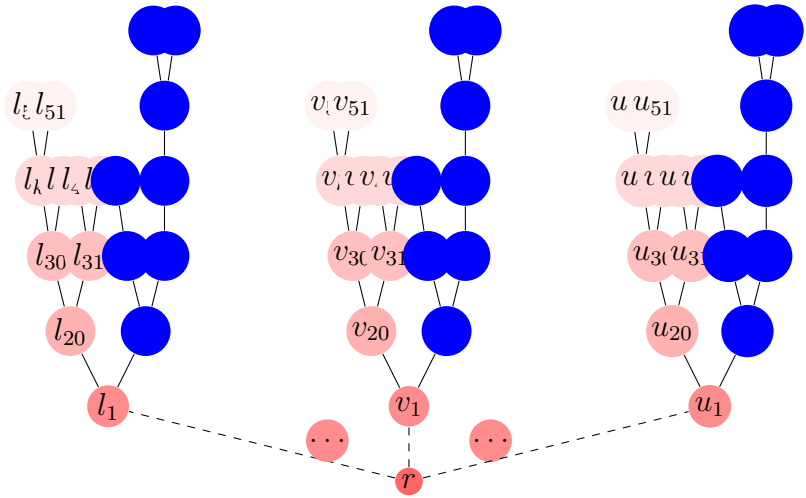
\begin{figure}[H]
    \begin{tikzpicture}
[level distance=10mm,
every node/.style={fill=red!60,circle,inner sep=1pt},
level 1/.style={sibling distance=20mm,nodes={fill=red!45}},
level 2/.style={sibling distance=10mm,nodes={fill=red!30}},
level 3/.style={sibling distance=5mm,nodes={fill=red!25}},
level 4/.style={sibling distance=3mm,nodes={fill=red!15}},
level 5/.style={sibling distance=3mm,nodes={fill=red!5}}]
\node {$r$} [grow'=up]
child[dashed] {node {$l_{1}$} 
child[solid] {node {$l_{20}$}
child {node {$l_{30}$}
  child {node {$l_{k0}$}
    child {node {$l_{50}$}}
    child {node {$l_{51}$}}
  }
child {node {$l_{41}$}}
}
child {node {$l_{31}$}
child {node {$l_{42}$}}
child {node {$l_{43}$}}
}
}
child[solid] {node [color=blue] {$l_{21}$}
child {node [color=blue] {$l_{32}$}
child {node [color=blue] {$l_{44}$}}
child[missing]
}
child {node [color=blue] {$l_{33}$}
  child {node [color=blue]{$l_{45}$}
    child {node [color=blue]{$l_{52}$}
      child {node [color=blue]{$l_{60}$}}
      child {node [color=blue]{$l_{61}$}}
      }}}}}
child { edge from parent[draw=none] node[draw=none] (ellipsis) {$\ldots$} }
child[dashed] {node {$v_{1}$} 
child[solid] {node  {$v_{20}$}
child {node {$v_{30}$}
  child {node {$v_{k0}$}
    child {node {$v_{50}$}}
    child {node {$v_{51}$}}
  }
child {node {$v_{41}$}}
}
child {node {$v_{31}$}
child {node {$v_{42}$}}
child {node {$v_{43}$}}
}
}
child[solid] {node [color=blue] {$v_{21}$}
child {node [color=blue] {$v_{32}$}
child {node (inputv1) [color=blue] {$v_{44}$}}
child[missing]
}
child {node [color=blue] {$v_{33}$}
  child {node (inputv2) [color=blue]{$v_{45}$}
    child {node [color=blue]{$v_{52}$}
      child {node [color=blue]{$v_{60}$}}
      child {node [color=blue]{$v_{61}$}}
      }
}}}}
child { edge from parent[draw=none] node[draw=none] (ellipsis) {$\ldots$} }
child[dashed] {node {$u_{1}$} 
child[solid] {node {$u_{20}$}
child {node {$u_{30}$}
  child {node {$u_{k0}$}
    child {node {$u_{50}$}}
    child {node {$u_{51}$}}
  }
child {node {$u_{41}$}}
}
child {node {$u_{31}$}
child {node {$u_{42}$}}
child {node {$u_{43}$}}
}
}
child[solid] {node [color=blue] {$u_{21}$}
child {node [color=blue] {$u_{32}$}
child {node [color=blue] {$u_{44}$}}
child[missing]
}
child {node [color=blue] {$u_{33}$}
  child {node [color=blue]{$u_{45}$}
    child {node [color=blue]{$u_{52}$}
      child {node [color=blue]{$u_{60}$}}
      child {node [color=blue]{$u_{61}$}}
}}}}};

       \end{tikzpicture}
    \caption{Some instances of the matrix $\mathcal{C}$ in the ambient r-DagProof $\mathcal{D}$}\label{fig:DetachLinkI}
\end{figure}

    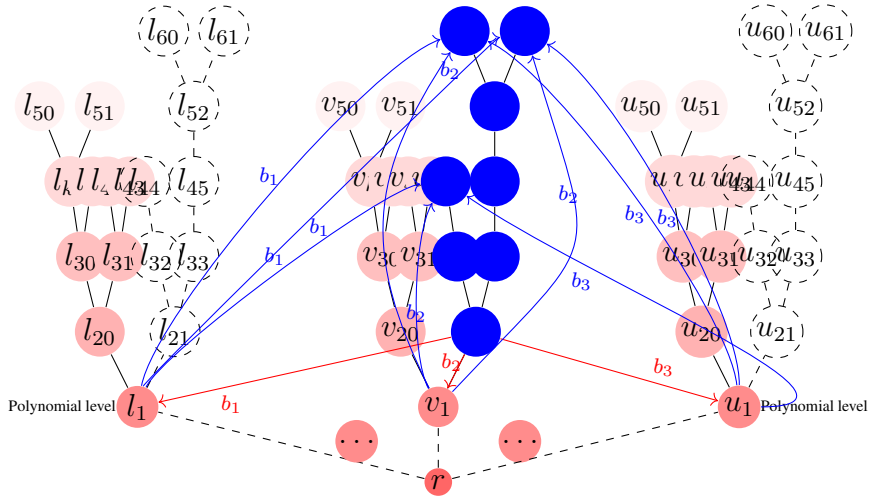
\begin{figure}[H]
    \begin{tikzpicture}
[level distance=10mm,
every node/.style={fill=red!60,circle,inner sep=1pt},
level 1/.style={sibling distance=20mm,nodes={fill=red!45}},
level 2/.style={sibling distance=10mm,nodes={fill=red!30}},
level 3/.style={sibling distance=5mm,nodes={fill=red!25}},
level 4/.style={sibling distance=3mm,nodes={fill=red!15}},
level 5/.style={sibling distance=8mm,nodes={fill=red!5}}]
\node {$r$} [grow'=up]
child[dashed] {node (entrada1) {$l_{1}$} 
child[solid] {node {$l_{20}$}
child {node {$l_{30}$}
  child {node {$l_{k0}$}
    child {node {$l_{50}$}}
    child {node {$l_{51}$}}
  }
child {node {$l_{41}$}}
}
child {node {$l_{31}$}
child {node {$l_{42}$}}
child {node {$l_{43}$}}
}
}
child[dashed] {node  [draw=black,dashed,fill={none}] {$l_{21}$}
child {node (sai11) [draw=black,dashed,fill={none}] {$l_{32}$}
child {node (ent11) [draw=black,dashed,fill={none}] {$l_{44}$}}
child[missing]
}
child {node [draw=black,dashed,fill={none}] {$l_{33}$}
  child {node [draw=black,dashed,fill={none}]{$l_{45}$}
    child {node [draw=black,dashed,fill={none}]{$l_{52}$}
      child {node [draw=black,dashed,fill={none}]{$l_{60}$}}
      child {node [draw=black,dashed,fill={none}]{$l_{61}$}}
      }}}}}
child { edge from parent[draw=none] node[draw=none] (ellipsis) {$\ldots$} }
child[dashed] {node (sai2) {$v_{1}$} 
child[solid] {node  {$v_{20}$}
child {node {$v_{30}$}
  child {node {$v_{k0}$}
    child {node {$v_{50}$}}
    child {node {$v_{51}$}}
  }
child {node {$v_{41}$}}
}
child {node {$v_{31}$}
child {node {$v_{42}$}}
child {node {$v_{43}$}}
}
}
child[solid] {node (saida) [color=blue] {$v_{21}$}
child {node [color=blue] {$v_{32}$}
child {node [color=blue] (ent21) {$v_{44}$}}
child[missing]
}
child {node (ent22) [color=blue] {$v_{33}$}
  child {node [color=blue]{$v_{45}$}
    child {node [color=blue]{$v_{52}$}
      child {node (ent60) [color=blue]{$v_{60}$}}
      child {node (ent61) [color=blue]{$v_{61}$}}
      }
}}}}
child { edge from parent[draw=none] node[draw=none] (ellipsis) {$\ldots$} }
child[dashed] {node (entrada2) {$u_{1}$} 
child[solid] {node {$u_{20}$}
child {node {$u_{30}$}
  child {node {$u_{k0}$}
    child {node {$u_{50}$}}
    child {node {$u_{51}$}}
  }
child {node {$u_{41}$}}
}
child {node {$u_{31}$}
child {node {$u_{42}$}}
child {node {$u_{43}$}}
}
}
child[dashed] {node [draw=black,dashed,fill={none}] {$u_{21}$}
child {node (sai31) [draw=black,dashed,fill={none}] {$u_{32}$}
child {node (ent31)[draw=black,dashed,fill={none}] {$u_{44}$}}
child[missing]
}
child {node [draw=black,dashed,fill={none}] {$u_{33}$}
  child {node [draw=black,dashed,fill={none}]{$u_{45}$}
    child {node [draw=black,dashed,fill={none}]{$u_{52}$}
      child {node [draw=black,dashed,fill={none}]{$u_{60}$}}
      child {node [draw=black,dashed,fill={none}]{$u_{61}$}}
}}}}};

\node (poly) [fill={none}, left of=entrada1] {{\tiny Polynomial level}} ;
\node (poly) [fill={none}, right of=entrada2] {{\tiny Polynomial level}} ;

%\draw[black,-] (poly) -- (entrada1) {};
\draw[red,->] (saida) -- (entrada2) node[draw=none,fill=none,near end, above, font=\scriptsize] {$b_3$};
\draw[red,->] (saida) -- (sai2) node[draw=none,fill=none, near end, above, font=\scriptsize] {$b_2$};
\draw[red,->] (saida) -- (entrada1) node[draw=none,fill=none,near end, below, font=\scriptsize] {$b_1$};
\draw[blue,->] (sai2) to[out=0,in=0,distance=0.5cm,bend left] node[draw=none,fill=none,midway,below,font=\scriptsize] {$b_2$} (ent21);
\draw[blue,->] (sai2) to[out=20,in=0,distance=3cm,bend right] node[draw=none,fill=none,near end,below,font=\scriptsize] {$b_2$} (ent61);
\draw[blue,->] (sai2) to[out=100,in=0,distance=2cm,bend left] node[draw=none,fill=none,at end,below,font=\scriptsize] {$b_2$} (ent60);
%\draw[blue,->] (entrada1) to[out=60,in=0,distance=1cm,bend left] node[draw=none,fill=none,midway,below,font=\scriptsize] {$[\cdots,b_1]$} (ent22);
\draw[blue,->] (entrada1) to[out=0,in=0,distance=0.1cm,bend left] node[draw=none,fill=none,midway,below,font=\scriptsize] {$b_1$} (ent61);
\draw[blue,->] (entrada1) to[out=0in=0,distance=1cm,bend left] node[draw=none,fill=none,midway,above,font=\scriptsize] {$b_1$} (ent60);
\draw[blue,->] (entrada1) to[out=60,in=0,distance=0.5cm,bend left] node[draw=none,fill=none,midway,above,font=\scriptsize] {$b_1$} (ent21);
\draw[blue,->] (entrada2) to[out=100,in=0,distance=1cm,bend right] node[draw=none,fill=none,midway,below,font=\scriptsize] {$b_3$} (ent60);
\draw[blue,->] (entrada2) to [out=0,in=-35] node[draw=none,fill=none,near end,below,font=\scriptsize] {$b_3$} (ent21);
\draw[blue,->] (entrada2) to[out=90,in=45,distance=1cm,bend right] node[draw=none,fill=none,midway,below,font=\scriptsize] {$b_3$} (ent61);
%\draw[black,->] (sai31) to[out=80,in=0,distance=1cm,bend right] node[draw=none,fill=none,midway,below,font=\scriptsize] {} (saida);
%\draw[black,->] (sai11) to[out=80,in=0,distance=1cm,bend left] node[draw=none,fill=none,midway,below,font=\scriptsize] {} (saida);
%\draw[black,->] (ent31) to[out=-100,in=0,distance=1cm,bend left] node[draw=none,fill=none,midway,below,font=\scriptsize] {} (saida);
%(e) to[out=135,in=180,distance=2.5cm] node{} (a)
\end{tikzpicture}
    \caption{ Three DetachLink were applied in the ambient r-DagProof $\mathcal{D}$ of fig.~\ref{fig:DetachLinkI}}\label{fig:DetachLinkII}
\end{figure}

        We remember that $r_{\mathcal{C}}$ is the root of the r-DagProof $\mathcal{C}$ and $Starts(\mathcal{C})$ is the set of initials of $\mathcal{C}$ as stated by definition~\ref{def:Initials}

     We observe that the resulting (pre) r-DagProof yielded from the collapse in level $\mu$ can be bigger than $m^p$ yet. With the collapse of $m^{p-3}$
     sub-proofs/derivations, the size of the resulting r-DagProof  is at least $\frac{\size{\mathcal{D}}}{\size{\mathcal{C}}\times m^{p-3}}$. If the mentioned size of the resulting r-DagProof is bigger than $m^{p}$, then there must be two main reasons: (1) The collapsed sub-proof/derivation is bigger then $m^p$ by itself, or; (2) There must be more matrices in level $\mu$ that we consider. The second alternative dealt proceeds by collapsing all instances of all matrices occurring at the lowest level. This is addressed in lines~\ref{line:Corollary} to~\ref{line:FimForCorollary} of algorithm~\ref{algo:CompressEmNDProof}. For the first alternative, we only have to recursively find more redundant parts in the sub-proof/derivation that must exist by theorem~\ref{main} in the matrix that had all instances collapsed in the ambient r-DagProof. Since theorem~\ref{main} works on {\bf EmND} proofs, as opposed to derivations with open assumptions, we need lemma~\ref{lemma:PrepRecursive} below.

     \begin{lemma}\label{lemma:PrepRecursive}
Let  $p\in\mathbb{N}$, $p>3$, $\Pi$ be an {\bf EmND} proof of a \mil tautology $\alpha$, such that, $\size{\mathcal{T}(c(\Pi))}=m$ and $\size{\Pi}>m^p$. According theorem~\ref{main} let $\Pi_{s}$ be a matrix that has at least $m^{p-3}$ instances occurring in the level $\mu$ in $\Pi$. If $\size{\Pi_{s}}>m^p$ then there is a matrix $\Pi_{s}^{\prime}$ that is a sub-derivation of $\Pi_{s}$, such that it has at least $m^{p-3}$ instances in $\Pi_{s}$ in some, fixed, level $\nu>\mu$.
       \end{lemma}

     \begin{proof} of~\ref{lemma:PrepRecursive}. The main feature of theorem~\ref{main} is the fact that it works for EmND proofs and $\Pi_{s}$ does not need to be a proof, it is instead a derivation with open assumptions. We can discharge in the correct order, dictated by the syntax tree of $\alpha$ and the mapping l that defines the EmND $\Pi_{s}$. To do this, we restrict l to $\Pi_{s}$. This restriction does not yield an EmND yet, and we have to consider the syntax tree of the conclusion of $\Pi_{s}$, i.e. the syntax tree $T_{c(\Pi_{s})}$. However, as we have already said, we do not have a tautology as the conclusion of $\Pi_{s}$. We discharge all open assumptions in $\Pi_{s}$, obtaining a conclusion $\beta$ that is a tautology. The result is then an EmND that proves $\beta$. The introduction part of the main branch does not disrupt the condition o definition~\ref{def:E-mapped-ND}.
       Finally, since $\beta$ is smaller than $\alpha$ we choose a new propositional variable $q_{new}$ and define a new formula ``$q_{new}\imply q_{new}\imply\ldots\beta$'' with as many $q_{new}$ repetitions as it is enough to have a formula of  the same size as $\alpha$. The EmND adjusted to this new formula is a proof of a tautology, and we can apply Theorem~\ref{main} to obtain a matrix $\Pi_{s}^{\prime}$ that has $m^{p-3}$ instances in $\Pi_{s}$. Note that $\size{T_{c(\Pi_{s})}} = \size{T_{a}}$. Since there is only one main branch, the many instances of $\Pi_{s}^{\prime}$ occur above the original conclusion of $\Pi_{s}$. Finally, the I-part of the main branch, from $c(\Pi_{s})$ down to ``$q_{new}\imply q_{new}\imply\ldots\beta$'' can be eliminated by merely deleting the introduction rules applied to draw ``$q_{new}\imply q_{new}\imply\ldots\beta$'', including the rules used to prove $\beta$.
       In this way we obtain the original matrix $\Pi_{s}$ and the desired $m^{p-3}$ instances of $\Pi_{s}^{\prime}$ that occurs in it. Finally, as $\Pi_{s}^{\prime}$ is sub-derivation of $\Pi_{s}$ then the level $\nu$ of its root is strictly above $\mu$.

       \end{proof}

By the proof of Lemma~\ref{lemma:PrepRecursive}, we can see that even in the case that $\Pi_{s}$ is not an {\bf EmND} proof, if $\size{\Pi_{s}}>m^p$, with $m=c(\Pi_{s})$, then there is a sub-derivation of $\Pi_s$ that it is repeated at least $m^{p-3}$ many times in $\Pi_{s}$.

     \begin{corollary}\label{corollary:main}
       Let  $p\in\mathbb{N}$, $p>3$, $\Pi$ be an {\bf EmND} derivation of $\alpha$, such that, $\size{\mathcal{T}(c(\Pi))}=m$ and $\Pi_s$ a sub-derivation of $\Pi$ that occurs in $\Pi$ in level $\mu$. If $\size{\Pi_{s}}>m^p$ then there is a sub-derivation $\Pi_{s}^{\prime}$ of $\Pi_{s}$ that has  at least $m^{p-3}$ instances occurring in a level $\nu>\mu$.
     \end{corollary}

     The above corollary~\ref{corollary:main} is used to ensure the termination and correctness of algorithm~\ref{algo:CompressEmNDProof} in the section~\ref{sec:SuccinctrDags}. Algorithm~\ref{algo:Collapse}, below,  defines the operation of collapsing a list $\mathcal{Y}$ of instances of the r-DagProof matrix $\mathcal{C}$\footnote{We remember that a matrix that occurs in a level $\mu$ in an r-DagProof $\mathcal{D}$ is any sub-r-DagProof of $\mathcal{D}$ that has root in level $\mu$ and may have some other instances in level $\mu$ too} occurring in level $\mu$

\begin{definition}[Collapse operation inside {\bf pre} r-DagProofs]
  Let
  \[
  \mathcal{D} = \langle V , E_{d}, E_{A}, r, l, L, \rho,\delta, \mathcal{O}_{\alpha}\rangle
  \]
  be a {\bf pre} r-DagProof for a \mil~ formula $\alpha$ and $\mathcal{C}$ be a {\bf pre} sub r-DagProof of $\mathcal{D}$ and $\mathcal{Y}$ the set of roots of the instances of $\mathcal{C}$, occuring in level $\mu$ that will be collapsed. Algorithm~\ref{algo:Collapse} defines and computes the result of collapsing all instances of $\mathcal{C}$ in only one in $\mathcal{D}$.  
\end{definition}

\begin{algorithm}[H]
%    \algsetup{linenosize=\tiny}
    \scriptsize
  \caption{}\label{algo:Collapse}
   \begin{algorithmic}[1]
    \REQUIRE{$\mathcal{D}$, $\mathcal{C}$, r-DagProofs, $\mathcal{C}\prec \mathcal{D}$,  and the list $\mathcal{Y}$ containing the roots of the instances of $\mathcal{C}$}  
    \ENSURE{a r-DagProof $\mathcal{D}^{\prime}$ having all instances of $\mathcal{C}$ collapsed in a unique r-DagProof}
    \STATE{{\bf Function}\;$Collapse(\mathcal{D},\mathcal{Y},\mathcal{C})$}
    \STATE{$\mathcal{D}^{\prime}\leftarrow\mathcal{D}$}
    \STATE{$j\leftarrow l_{\mathcal{C}}(root(\mathcal{C}));$}
    \STATE{$i\leftarrow 1$;}
    \STATE{$\mathcal{Y}\leftarrow rest(\mathcal{Y})$} 
    \FOR{$k\in\mathcal{Y}$}
       \STATE{$\mathcal{D}^{\prime}\leftarrow DetachLink(\mathcal{C},k,\mathcal{D}^{\prime},j@i); i\leftarrow i+1$}
    \ENDFOR
%    \STATE{FreeSlots $\leftarrow$ $\{v:\mbox{$v\in V(\mathcal{D})$ and $in(v)=0$}\}$}
%    \STATE{BoundSlots $\leftarrow$ $\{v:\mbox{$v\in V(\mathcal{D})$ and $\exists i\in\mathcal{O}_{B}$ $\langle r(D_{\mathcal{Y}}),v\rangle\in E_{D}^{i}$}\}$} 
    \STATE{$Return\;\; \mathcal{D}$} 
   \end{algorithmic}
\end{algorithm}

Examining algorithm~\ref{algo:Collapse}, we have the Lemma~\ref{lemma:SizeCollapse} that provides an upper-bounded for the resulting r-DagProof after the collapses of all instances of the matrix in the ambient r-DagProof $\mathcal{D}$. 

 \begin{lemma}\label{lemma:SizeCollapse}
   $\size{\mathcal{D}^{\prime}} \leq \frac{\size{\mathcal{D}}}{length(\mathcal{Y})\times\size{\mathcal{C}}}$
  \end{lemma}

\section{r-Dags as succinct certificates for \mil~ tautologies}\label{sec:SuccinctrDags}

In this section, we show how to use the {\bf Collapse} operation defined in previous section~\ref{sec:RemovingRedundancies}. Algorithm~\ref{algo:CompressEmNDProof} defines the operation of compression that collapses all redundancies that occur in any huge {\bf EmND} proof. In line~\ref{line:Corollary} the function $Lemma\ref{lemma:ListForCollapse}(\mathcal{T}))$ returns the set of independent matrices in $\mathcal{T}$ that exists by Lemma~\ref{lemma:ListForCollapse}. Since it is an independent set, by collapsing all of its instances in only one instance to each matrix, we do not need to collapse the upper levels that are not local lowest levels, the roots of the matrices' elements in the independent set.

%\begin{definition}[Checking {\bf pre} rDagProofs for the validity of its conclusion]
 \begin{algorithm}[H]
%    \algsetup{linenosize=\tiny}
    \scriptsize
  \caption{Compress an {\bf EmND} proof $\mathcal{T}$ using corol. }\label{algo:CompressEmNDProof}
  \begin{algorithmic}[1]
    \REQUIRE{Uses the global variable $m$ with value $\size{c(\mathcal{T})}$ }
    \REQUIRE{$3<p\in \mathbb{N}$, $\mathcal{T}$ is a {\bf EmND},  $h(T)\leq \size{c(\mathcal{T})}$}
    \ENSURE{a r-DagProof $\mathcal{D}$ proving $c(\mathcal{T})$ of size smaller than $\size{c(\mathcal{T})}^{p-3}$}
    \STATE{{\bf Function}\;Compress($\mathcal{T},p$)}     
    \IF{($m^{p}< \size{\mathcal{T}})$}\label{line:IF}
    \STATE{$LocalLowestLevels\leftarrow MinLevel(Lemma~\ref{lemma:ListForCollapse}(\mathcal{T}))$}\label{line:Corollary}
    \STATE{$\mathcal{D}\leftarrow \mathcal{T}$}
    \FOR{$lev\in LocalLowestLevels$ upwards $h(\mathcal{T})$}
    \STATE{$L\leftarrow SuperPolySubProofs(\mathcal{T},lev)$}
    \FOR{$\langle\mathcal{Y},\mathcal{C}_{\mathcal{Y}}\rangle \in L$}
    \STATE{$\mathcal{D}_{\mathcal{Y}}\leftarrow Compress(\mathcal{C}_{\mathcal{Y}},p)$}
    \STATE{$\mathcal{D}\leftarrow Collapse(\mathcal{D},\mathcal{Y},D_{\mathcal{Y}})$}
    \ENDFOR
    \ENDFOR\label{line:FimForCorollary}
%    \STATE{FreeSlots $\leftarrow$ $\{v:\mbox{$v\in V(\mathcal{D})$ and $in(v)=0$}\}$}
%    \STATE{BoundSlots $\leftarrow$ $\{v:\mbox{$v\in V(\mathcal{D})$ and $\exists i\in\mathcal{O}_{B}$ $\langle r(D_{\mathcal{Y}}),v\rangle\in E_{D}^{i}$}\}$} 
    \STATE{$Return \mathcal{D}$} 
    \ELSE
    \STATE{$Return \mathcal{T}$}
    \ENDIF
   \end{algorithmic}
\end{algorithm}

 %  \end{definition}

 Below we have some lemmas that are proved easily by inspecting the code of algorithms~\ref{algo:Collapse},~\ref{algo:CompressEmNDProof} and definition~\ref{def:Detach}.

\begin{lemma}\label{lemma:MaximalCallsCompressProofs}

  In algorithm~\ref{algo:CompressEmNDProof} above, the number of recursive calls after an initial invocation of $Compress(\mathcal{T},p)$, with $m=\size{c(\mathcal{T})}$, is at most $\size{\mathcal{T}}$.

\end{lemma}

\begin{proof} Since $m^p$ is constant during all the recursive calls, and the size of the first argument of $Compress$ is strictly smaller than each previous recursive call.  Thus, there must be a call such that $\size{\mathcal{T}}<m^p$. When $\mathcal{T}$ is of this size, it is the recursive basis case, anyway.
  Taking the (worst) case into account, the number of recursive calls is upper-bounded by the size of the EmND $\mathcal{T}$ itself. If any recursive call, on the first argument $\mathcal{T}_1$, is such that $\size{\mathcal{T}_1}>m^p$, then it obtains the sub-derivations that occurs at least $m^{p-3}$, collapsing $Compress(\mathcal{T}_1)$ into a unique compressed version of $\mathcal{T}_1$. This is done with lesser  than $\size{\mathcal{T}}$ recursive calls. In fact, there is at least $m^{p-3}$ collapses, so the total amount of recursive calls is at most $\frac{\size{\mathcal{T}}}{\size{\mathcal{T}_1}\times m^p}$. Lastly, the compression of $\mathcal{T}_1$ is obtained in a unique call to $Compress$.
\end{proof}

Below we find a useful Lemma used to prove that after the compression any linear height bounded {\bf EmND} super-polynomially bounded proof becomes polynomial in size.

\begin{definition}\label{def:LRI} Given a {\bf EmND} linearly hight bounded proof $\mathcal{T}$, such that, $\size{\mathcal{T}}>\size{\alpha}^p$, $p\ge 3$.  We denote by  $LRI(\mathcal{T})$\footnote{$LRI$ is an acronym to $LowestInstancesRedundant$} the list of all lowest instances that occur at least $m^{p-3}$ times, for each of the matrices, as provided by Lemma~\ref{lemma:ListForCollapse}. 
\end{definition}

\begin{definition}[EmND proofs difference operation]\label{def:diff}
  Let $\mathcal{T}$ and $\mathcal{T}_1$ EmND derivations, such that, $\mathcal{T}_1$ is sub-derivation of $\mathcal{T}$ occuring in level $\nu$. The difference derivation $\mathcal{T}^{\prime}=\mathcal{T}-\mathcal{T}_1$ obtains by removing all nodes of $\mathcal{T}_1$ from $\mathcal{T}$, but $c(\mathcal{T}_1)$. Morevore, $l_{\mathcal{T}^{\prime}}$ is the restriction of $l_{\mathcal{T}}$ to this $\mathcal{T}^{\prime}$ new derivation.  
\end{definition}

\begin{definition}[EmND proof difference with removal of conclusion]\label{def:ominus}
  Let $\mathcal{T}$ and $\mathcal{T}_1$ EmND derivations, such that, $\mathcal{T}_1$ is sub-derivation of $\mathcal{T}$ occuring in level $\nu$. The greedy difference $\mathcal{T}\ominus\mathcal{T}_1$ is as in defintion~\ref{def:diff} extended with the removal of the conclusion of $\mathcal{T}_1$ from the yielded sub-derivation.
  \end{definition}

The difference of  $\mathcal{T}$ and the set of derivations $S$ is the iterating the difference of $\mathcal{T}$ to each member of $S$. 

The following Lemma~\ref{lemma:PolyPart} shows that for every proof of size bigger than $m^p$, when removing its redundant part, that exists in virtue of Theorem~\ref{main}, what remains is a rDagProof of size less than $m^p$.

We remember that derivations formally are labeled trees, and the later are graphs. The differnce between  graphs is well-defined as defined by definition~\ref{def:diff}

We need a preparatory lemma that is, in a certain sense, a generalization of Lemma~\ref{lemma:PrepRecursive}.

We use $\left(\begin{array}{c}\gamma_1\\\vdots\\\gamma_n\end{array}\right)\imply\alpha$ to denote the formula $\gamma_1\imply(\gamma_2\imply(\ldots\imply \alpha)\ldots)$. 

\begin{lemma}[Redundancy Lemma for Normal derivations]\label{lemma:RedundantDerivation}
Let $\Pi$ be a normal and expanded derivation in N.D. for $M_{\imply}$. Let $\Gamma=\{\gamma_1,\ldots,\gamma_n\}$ be the set of all open assumptions in $\Pi$ and $\alpha$ the conclusion of $\Pi$, $m=\size{\alpha}+\sum_{\gamma\in\Gamma}\size{\gamma}$. If $\size{\Pi}>m^p$, for some $p> 3$, then there is a sub-derivation $\Pi_s$ of $\Pi$ and a level $\mu$ from $\Pi$, such that, there are at least $m^{p-3}$ instances of $\Pi_s$ occurring in level $\mu$. 
  \end{lemma}

\begin{proof}
From $\Pi$ we obtain an expanded normal proof $\Pi^{\prime}$ of $\left(\begin{array}{c}\gamma_1\\\vdots\\\gamma_n\end{array}\right)\imply\alpha$ by applying a series $\imply$-I rules in the main branch of $\Pi$, discharging all open assumptions in $\Pi$. From the parsing tree of  $\left(\begin{array}{c}\gamma_1\\\vdots\\\gamma_n\end{array}\right)\imply\alpha$ and an adequate partial map from formula occurrences in $\Pi^{\prime}$ to this parsing tree we obtain an EmND proof of $\left(\begin{array}{c}\gamma_1\\\vdots\\\gamma_n\end{array}\right)\imply\alpha$. Note that any adequate mapping produces an EmND proof of the later formula. We apply Lemma~\ref{main} to $\Pi^{\prime}$ concluding that there is a level $\mu$ from $\Pi^{\prime}$ and a sub-derivation $\Pi_s$ of $\Pi^{\prime}$, such that, there are at least $m^{p-3}$ instances of $\Pi_s$ occurring in level $\mu$ in $\Pi^{\prime}$. Finally, we rebuild the orginal derivation $\Pi$ by removing all $\imply$-I applications used to construct $\Pi^{\prime}$ from $\Pi$. We reach the desire conclusion, i.e., there is a level $\mu$ from $\Pi$ and a sub-derivation $\Pi_s$ of $\Pi$, such that, there are at least $m^{p-3}$ instances of $\Pi_s$ occurring in level $\mu$ in $\Pi$
\end{proof}

Other very useful lemma is the following. It is inspired in the huge proofs, based on Fibonacci numbers, shown in \cite{SuperExp} in its appendix. 

\begin{lemma}\label{lemma:HugeDerivation}
  Let $\beta$ be a $M_{\imply}$ formula. There is a family $(\Pi^{fib}_{n})_{n\in\mathbb{N}})$ of $M_{\imply}$ N.D.  derivations of $\beta$ from a set of formulas $\Delta_n$, of linear size on $n$ and $\size{\beta}$. If $m=\size{\beta}+\sum_{\delta\in\Delta_n}\size{\delta}$ then, for each $n\in\mathbb{N}$, $\Pi^{fib}_{n}$ is a Normal and expanded derivation of height $\mathcal{O}(4\times n\times\size{\beta})$ and size lower-bounded by $\frac{\phi^{\size{\beta}\times n}}{\sqrt{5}}$, where $\phi\approx 1.618$ is the golden ratio.
\end{lemma}

\begin{proof} In \cite{Exponential}, section 4, we show a family of Natural Deduction $M_{\imply}$ derivations of $p_1\imply p_n$, from $\Delta_n=\{p_1,p_1\imply p_2\}\cup\{p_i\imply(p_{i+1}\imply p_{i+2}):i\leq n-2\}$ of height $n$ and size $\frac{\phi^{ n}}{\sqrt{5}}$. We modify this family to have $p_n=\beta$, getting the statement of the lemma.
  \end{proof}

  \begin{lemma}\label{lemma:PolyPart}
    For any EmND proof $\mathcal{T}$ of a tautology $\alpha$ and  $p>3$, if $\size{\mathcal{T}}\ge\size{\alpha}^p$ and $\langle\mathcal{T}_1,\ldots,\mathcal{T}_n\rangle=LRI(\mathcal{T})$, then 
    \[
    \size{\mathcal{T}\ominus\bigcup_{i=1\ldots n}\mathcal{T}_i}< m^p
    \]
  \end{lemma}

  \begin{proof}
    We can consider $\mathcal{T}$ as shown in figure~\ref{figure:T}, where $\ell\times m^{p-3}\leq n$, and $\ell$ is the number of different levels that have sub-derivations occurring at least $m^{p-3}$ in them. We remember that the formulas $\beta_i$, $i=1\ldots n$ may occur in different levels. For each $j=1\ldots \ell$ there are at least $m^{p-3}$ occurrences of instances of a sub-derivation $\mathcal{T}_j$ in level $j$ in $\mathcal{T}$. To facilitate the understanding we re-index the instances with the level $j$, such that, $\mathcal{T}_{j_i}^j$ is the $i$-th instance  of the matrix $\mathcal{M}_j$ occurring in level $j$ in $\mathcal{T}$. We remember that there are $\ell$ matrices indicating $\ell$ redundant parts in $\mathcal{T}$.
    \begin{figure}[H]
      \begin{prooftree}
        \AxiomC{$\mathcal{T}_1^{\mu_1}$}
        \noLine
        \UnaryInfC{$\beta_1^{\mu_1}$}
        \AxiomC{$\delta_1^{\mu_1}$}
        \noLine
        \UnaryInfC{$\Pi_1^{\mu_1}$}
        \noLine
        \UnaryInfC{$\beta_1^{\mu_1}\imply\gamma_1^{\mu_1}$}
        \BinaryInfC{$\gamma_1^{\mu_1}$}
        \AxiomC{$\ldots$}
        \AxiomC{$\mathcal{T}_1^{\mu_\ell}$}
        \noLine
        \UnaryInfC{$\beta_1^{\mu_\ell}$}
        \AxiomC{$\delta_1^{\mu_\ell}$}
        \noLine
        \UnaryInfC{$\Pi_1^{\mu_\ell}$}
        \noLine
        \UnaryInfC{$\beta_1^{\mu_\ell}\imply\gamma_1^{\mu_\ell}$}
        \BinaryInfC{$\gamma_1^{\mu_\ell}$}
        \noLine
        \TrinaryInfC{$\Pi$}
         \noLine
         \UnaryInfC{$q$}
         \noLine
         \UnaryInfC{$\vdots$}
         \noLine
         \UnaryInfC{$\alpha$}
      \end{prooftree}
      \caption{The proof $\mathcal{T}$ in Lemma~\ref{lemma:PolyPart}}\label{figure:T}
    \end{figure}

    We prove by induction on the number of different lowest levels, i.e, on $\ell$ that
    \[
        \size{\mathcal{T}_0}=\size{\mathcal{T}\ominus\bigcup_{i=1\ldots n}\mathcal{T}_i}< m^p
        \]

        Induction on $\ell$:
    \begin{enumerate}
    \item[Basis] There is only one level $\mu$ and one matrix $\mathcal{T}_{\mu}$, hence $\ell=1$,  that is repeated at least $m^{p-3}$ times in level $\mu$ in $\mathcal{T}$.  Let $\mathcal{T}_{0}=\mathcal{T}\ominus\bigcup_{i=1\ldots n}\mathcal{T}_i$, where $\mathcal{T}_i^{\mu_i}$ are all the instances of $\mathcal{T}_{\mu}$, for all $i=1\ldots n$. We notice that $\mathcal{T}_0$ is not a derivation anymore, it is only a tree, depicted in figure~\ref{figure:T0}.
     \begin{figure}[H]
      \begin{prooftree}
 %       \AxiomC{$\mathcal{T}_1^{\mu_1}$}
 %       \noLine
 %       \UnaryInfC{$\beta_1^{\mu_1}$}
        \AxiomC{$\delta_1^{\mu_1}$}
        \noLine
        \UnaryInfC{$\Pi_1^{\mu_1}$}
        \noLine
        \UnaryInfC{$\beta_1^{\mu_1}\imply\gamma_1^{\mu_1}$}
        \UnaryInfC{$\gamma_1^{\mu_1}$}
        \AxiomC{$\ldots$}
%        \AxiomC{$\mathcal{T}_1^{\mu_\ell}$}
%        \noLine
%        \UnaryInfC{$\beta_1^{\mu_\ell}$}
        \AxiomC{$\delta_n^{\mu_1}$}
        \noLine
        \UnaryInfC{$\Pi_n^{\mu_1}$}
        \noLine
        \UnaryInfC{$\beta_n^{\mu_1}\imply\gamma_n^{\mu_1}$}
        \UnaryInfC{$\gamma_n^{\mu_1}$}
        \noLine
        \TrinaryInfC{$\Pi$}
         \noLine
         \UnaryInfC{$q$}
         \noLine
         \UnaryInfC{$\vdots$}
         \noLine
         \UnaryInfC{$\alpha$}
      \end{prooftree}
      \caption{The tree $\mathcal{T}_0$}\label{figure:T0}
    \end{figure}

      Let us suppose that: 
    \[
    \size{\mathcal{T}_0}=\size{\mathcal{T}\ominus\bigcup_{i=1\ldots n}\mathcal{T}_i}\ge m^p
    \]
    We observe that by Lemma~\ref{lemma:RedundantDerivation}, for each $i$,  $\size{\Pi_{i}^{\mu_i}}<m^p$ and $\size{\Pi}<m^p$. If any of these derivations were bigger than $m^p$ there would be a matrix occurring at least $m^{p-3}$ in some level of them. The only instances that occur at least $m^{p-3}$ in $\mathcal{T}$ are $\mathcal{T}_i$, $i=1,\ldots,n$ by hypotheses. They occur in level $\mu$, contradicting the above assumptions on the size of the derivations $\Pi_{i}^{\mu_i}$ and $\Pi$. We note that $\size{\mathcal{T}_{0}}=\size{\Pi}+\sum_{i=1,n}\size{\Pi_i^{\mu_i}}\ge m^p$ and, by previous observation on the size of each of the summands, there must be a sub-sequence  $(\Pi_{r_i}^{\mu_{1}})_{i=1,r}$ of the sequence $(\Pi_{i}^{\mu_1})_{i=1,n}$ of derivations, such that, $\size{\Pi}+\sum_{i=1,r}\size{\Pi_{r_i}^{\mu_{1}}}\ge m^p$. We have to consider two cases:
    \begin{description}
    \item[Proper] We have that $r<m^{p-3}$. So, we build the derivation $\mathcal{T}_0^{++}$ in figure~\ref{figure:T0++}, by re-introducing in $\mathcal{T}_0$ the respective minor premiss $\beta_{r_i}^{\mu_{1}}$ for each $\Pi_{r_i}^{\mu_{1}}$ derivation of the major premiss $\beta_{r_i}^{\mu_{1}}\imply\gamma_{r_i}^{\mu_{1}}$. Moreover, the other derivations $\Pi_i^{\mu_1}$ not in the sub-sequence $(\Pi_{r_i}^{\mu_{1}})_{i=1,r}$ are erased from $\mathcal{T}_0$. We re-inforce the observation that all $\beta_{i}^{\mu_1}$ are the same formula.Moreover, we delete the $I-part$ of the main branch of $\mathcal{T}$ that occurs in $\mathcal{T}_0^{++}$. This is the final derivation $\mathcal{T}_0^{++}$ show in the figure~\ref{figure:T0++}. $q$ is the minimal formula of the main branch.

%% Estou AQUI

   \begin{figure}[H]
      \begin{prooftree}
 %       \AxiomC{$\mathcal{T}_1^{\mu_1}$}
 %       \noLine
 %       \UnaryInfC{$\beta_1^{\mu_1}$}
 %       \AxiomC{$\delta_1^{\mu_1}$}
 %       \noLine
 %       \UnaryInfC{$\Pi_1^{\mu_1}$}
 %       \noLine
 %       \UnaryInfC{$\beta_1^{\mu_1}\imply\gamma_1^{\mu_1}$}
        \AxiomC{$\gamma_1^{\mu_1}\;\;\ldots$}
 %       \AxiomC{$\ldots$}
%        \AxiomC{$\mathcal{T}_1^{\mu_\ell}$}
%        \noLine
        \AxiomC{$\beta_{r_1}^{\mu_{1}}$}
        \AxiomC{$\delta_{r_1}^{\mu_{1}}$}
        \noLine
        \UnaryInfC{$\Pi_{r_1}^{\mu_{1}}$}
        \noLine
        \UnaryInfC{$\beta_{r_1}^{\mu_{1}}\imply\gamma_{r_1}^{\mu_{1}}$}
        \BinaryInfC{$\gamma_{r_1}^{\mu_{1}}$}
        \AxiomC{$\ldots\;\;\gamma_{i}^{\mu_{1}}\;\;\ldots$}
        \AxiomC{$\beta_{r_r}^{\mu_{1}}$}
        \AxiomC{$\delta_{r_r}^{\mu_{1}}$}
        \noLine
        \UnaryInfC{$\Pi_{r_r}^{\mu_{1}}$}
        \noLine
        \UnaryInfC{$\beta_{r_r}^{\mu_{1}}\imply\gamma_{r_r}^{\mu_{1}}$}
        \BinaryInfC{$\gamma_{r_r}^{\mu_{1}}$}
        \AxiomC{$\ldots\;\;\gamma_{n}^{\mu_{1}}$}       
        \noLine
        \QuinaryInfC{$\Pi$}
         \noLine
         \UnaryInfC{$q$}
%         \noLine
%         \UnaryInfC{$\vdots$}
%         \noLine
%         \UnaryInfC{$\alpha$}
      \end{prooftree}
      \caption{The derivation $\mathcal{T}_0^{++}$}\label{figure:T0++}
    \end{figure}
   We must observe that $\mathcal{T}_0^{++}$ is a valid derivation bigger than $m^p$, thus by Lemma~\ref{lemma:RedundantDerivation} there is a level $\nu$, such that, there are at least $m^{p-3}$ instances of the matrix $\mathcal{M}$. We have a contradiction.
   \begin{itemize}
   \item If $\nu>\mu$ then there is at least $m^{p-3}$ instances of a derivation in a level that exists in all the sub-derivations $(\Pi_{r_i}^{\mu_{1}})_{i=1,r}$. Replacing back all the instances $\mathcal{T}_i$, $i=1,n$, we obtain $\mathcal{T}$ and a proof that it has another level, besides $\mu$, that has at least $m^{p-3}$ repeated instances of the same matrix. We have a contradiction.
   \item If $\nu=\mu$, since $\beta=\beta_i^{\mu_1}$, for each $i$, occurs less than $m^{p-3}$ in level $\mu$, then there must be other instances than $\mathcal{T}_i$ occurring at least $m^{p-3}$ in $\mathcal{T}$ in $\mu$. We have a contradiction, when we replace back the derivations $\mathcal{T}_i$, for there will be more than one matrix with at least $m^{p-3}$ repeated occurrences in level $\mu=\nu$.
   \item If $\nu<\mu$ we have a contradiction, by a reasoning analogous to the previous item.
   \end{itemize}
 \item[Improper] We have that $m^{p-3}\leq r$. So, we build the derivation $\mathcal{T}_0^{+++}$ in figure~\ref{figure:T0+++}, by re-introducing in $\mathcal{T}_0$ the respective minor premiss $\beta$ as in $\mathcal{T}_0^{++}$, but only from $r_1$ to $r_{m^{p-3}}$. Moreover, for each $i=m^{p-3}+1,\ldots,r$ we remove $\Pi_{r_i}^{\mu}$ from $\mathcal{T}_0$ and replace $\beta_{i-m^{p-3}}^{\mu}$, the derivation of the minor premiss of the $\imply$-Elim rule by the smallest derivation of the form stated by Lemma~\ref{lemma:HugeDerivation} that is bigger than $\size{\Pi_{r_i}^{\mu}}$. We denote by $\Sigma_{\Pi_{r_i}^{\mu}}$ this derivation. Note that the size of $\Sigma_{\Pi_{r_i}^{\mu}}$ compensates the removal of $\Pi_{i}^{\mu_1}$, for every $i=m^{p-3}+1$ to $i=r$. It is important to note that the height of $\Sigma_{r_i}^{\mu}$ is linear on $m$, so it is $\mathcal{T}_0^{+++}$. From $i=r+1$ to $i=n$ there is only a simply removal of $\Pi_{r_i}^{\mu_1}$, without any compensation. At the end the size of $\mathcal{T}_0^{+++}$ is bigger than $m^p$ and there is less than $m^{p-3}$ occurrences of $\beta$ in the level $\mu_1=\mu$. Thus, by an analysis similar to the above case, with the application of Lemma~\ref{lemma:RedundantDerivation}, we obtain a contradiction. We observe that we introduced a linear number of assumptions in the  Lemma~\ref{lemma:HugeDerivation}. We notice that we reason in terms of the big-O notation, thus we have that $\mathcal{T}_{0}^{+++}$ is bigger than $m^p$, in terms of $m$. 
   \begin{figure}[H]
      \begin{prooftree}
 %       \AxiomC{$\mathcal{T}_1^{\mu_1}$}
 %       \noLine
 %       \UnaryInfC{$\beta_1^{\mu_1}$}
 %       \AxiomC{$\delta_1^{\mu_1}$}
 %       \noLine
 %       \UnaryInfC{$\Pi_1^{\mu_1}$}
 %       \noLine
 %       \UnaryInfC{$\beta_1^{\mu_1}\imply\gamma_1^{\mu_1}$}
        \AxiomC{$\gamma_1^{\mu_1}\;\;\ldots$}
 %       \AxiomC{$\ldots$}
%        \AxiomC{$\mathcal{T}_1^{\mu_\ell}$}
        %        \noLine
        \AxiomC{$\Sigma_{r_1}^{\mu_{1}}$}
        \noLine
        \UnaryInfC{$\beta_{r_1}^{\mu_{1}}$}
        \AxiomC{$\delta_{r_1}^{\mu_{1}}$}
        \noLine
        \UnaryInfC{$\Pi_{r_1}^{\mu_{1}}$}
        \noLine
        \UnaryInfC{$\beta_{r_1}^{\mu_{1}}\imply\gamma_{r_1}^{\mu_{1}}$}
        \BinaryInfC{$\gamma_{r_1}^{\mu_{1}}$}
        \AxiomC{$\ldots\;\;\gamma_{i}^{\mu_{1}}\;\;\ldots$}
        \AxiomC{$\beta_{r_r}^{\mu_{1}}$}
        \AxiomC{$\delta_{r_r}^{\mu_{1}}$}
        \noLine
        \UnaryInfC{$\Pi_{r_r}^{\mu_{1}}$}
        \noLine
        \UnaryInfC{$\beta_{r_r}^{\mu_{1}}\imply\gamma_{r_r}^{\mu_{1}}$}
        \BinaryInfC{$\gamma_{r_r}^{\mu_{1}}$}
        \AxiomC{$\ldots\;\;\gamma_{n}^{\mu_{1}}$}       
        \noLine
        \QuinaryInfC{$\Pi$}
         \noLine
         \UnaryInfC{$q$}
%         \noLine
%         \UnaryInfC{$\vdots$}
%         \noLine
%         \UnaryInfC{$\alpha$}
      \end{prooftree}
      \caption{The derivation $\mathcal{T}_0^{+++}$}\label{figure:T0+++}
    \end{figure}
   
    \end{description}
  \item[Inductive Step] In this case we have $\ell>1$ levels with repetitions of the respective sub-derivations at least $m^{p-3}$. We analyse the highest level of repetitions in a similar way that was done in the basis step, with the hypothesis that the removal of all repetitions in the lowest $\ell-1$ levels result in a tree of size lesser than $m^p$. With this inductive hypothesis and the reasonning on the repetitions removal in the highest level we reach the desired conclusion.
    \end{enumerate}
     \end{proof}

\begin{lemma}\label{lemma:SizeCompressed}
  For any EmND proof $\mathcal{T}$ of a tautology $\alpha$ and  $p>3$, if $\size{\mathcal{T}}\ge\size{\alpha}^p$ then $\size{Compress(\mathcal{T},p)}<\size{\alpha}^p$
\end{lemma}

\begin{proof}
  We prove by induction on the number of recursive calls that $\size{Compress(\mathcal{T}, p)}<\size{\alpha}^p$, for in lemma~\ref{lemma:MaximalCallsCompressProofs} we have already proven that the algorithm stops for any  valid input pair $\langle \mathcal{T}, p\rangle$.
  \begin{itemize}
  \item {\bf Basis} No recursive call: In this case we already have $\size{Compress(\mathcal{T}, p)}<\size{\alpha}^p$. The ``{\bf else}'' of the ``if'' in line~\ref{line:IF} of algorithm~\ref{algo:CompressEmNDProof} is used. 
  \item {\bf I.H.} Suppose $\size{\mathcal{T}}\ge\size{\alpha}^p$ holds. Thus, a call to $Lemma~\ref{lemma:ListForCollapse}(\mathcal{T})$, in line~\ref{line:Corollary} return the list of all occurrences of sub-derivations, instances of the independent set of matrices given by Lemma~\ref{lemma:ListForCollapse}, that occurs more than $m^{p-3}$ in the lowest levels of $\mathcal{T}$. By inductive hypothesis, there is less one recursive call,  $Compress(\mathcal{T}_i)$ of each instance $\mathcal{T}_i$ of the list returns a r-Dag of size less than  $\size{\alpha}^p$. By Lemma~\ref{lemma:PolyPart} the part of $\mathcal{T}$ that it is not collapsed is less than $\size{\alpha}^p$ and by Lemma~\ref{lemma:SizeCollapse} we obtain that $\size{Compress(\mathcal{T}, p)}<\size{\alpha}^p$. 
    \end{itemize}

  \end{proof}

The above upper-bound is not tight. A tighter one obtains by counting the number of matrices for each level, according to Theorem~\ref{main} and corollary~\ref{corollary:main}. However,  in this article, we do not need a tighter upper-bound than what we state in lemma~\ref{lemma:MaximalCallsCompressProofs}.

\section{Checking r-DagProofs in Polynomial Time}\label{sec:PolytimeCompression}

The following definitions are central in proving that r-DagProofs are certificates for $\mil$ formulas validity. Let $\mathcal{C}$ = $\langle V , E_{d}, E_{A}, r, l, L,\rho,\delta,\mathcal{O}_{\alpha}\rangle$ be a {\bf pre} r-DagProof        of $\alpha$ from $\mil$. We associate to each $v\in V$ an entailment relation. The entailment represents the logical consequency relation carried within $\mathcal{C}$ from the DAG's leaves downwards until $v$. Due to the collapse operation, and the many downward detours that the collapses introduce, we use an environment function that keeps track of the entailment relation related to each detour. We use the notation $Emt(\alpha)$ to denote the set $\{\Delta\vdash\beta:\mbox{$\Delta\subseteq Sub(\alpha)$ and $\beta\in Sub(\alpha)$}\}$ of all possible entailments between sets of (sub)formulas of $\alpha$, $\Delta$, and (sub)formulas of $\alpha$.

\begin{definition}
  Let $\Delta\vdash\delta\in Emt(\alpha)$, and $\mathcal{O}$ be a total order on the subformulas of $alpha$. We define $Ant_{\mathcal{O}}(\Delta\vdash\delta)=b_{\mathcal{O}}(\Delta)$.
\end{definition}

In what follows, consider a  {\bf pre} r-DagProof $\mathcal{C}$ = $\langle V , E_{d}, E_{A}, r, \ell, L,\rho,\delta,\mathcal{O}_{\alpha}\rangle$ of $\alpha$. A node $v\in V$ is called a deductive leaf, iff, it has no incoming deductive edge, otherwise we call it as deductive (internal) node. The nodes of $\mathcal{C}$ that have mode than one different Deductive Edges outcoming from it are called divergent node.  

\begin{definition}[Local Entailment]\label{def:LocalEntailment}
  Given a {\bf pre} r-DagProof $\mathcal{C}$ = $\langle V , E_{d}, E_{A}, r, \ell, L,\rho,\delta,\mathcal{O}_{\alpha}\rangle$ of $\alpha$, we define the  mapping $M_{\vdash}^{\mathcal{C}}:V\times \mathbb{N}\longrightarrow Emt(\alpha)\cup\{\symbvdash\}$, for each $v\in V$ recursively as follows:
  \begin{description}
  \item [{\color{blue} Deductive Leaf, no Ancestrality}] If  $v\in V$ and, there is no $u\in V$, such that $\langle u,v\rangle\in E_{d}$ and, there is no $w\in V$, such that $\langle w,v\rangle\in E_{A}$ then $M_{\vdash}^{\mathcal{C}}(v,0)=\{\ell(v)\}\vdash \ell(v)$, and $M_{\vdash}^{\mathcal{C}}(v,j)=\symbvdash$, for $j\in\mathbb{N}$, $j\neq 0$, and;
  \item [{\color{blue} Deductive Leaf with Ancestrality}]If  $v\in V$ and, there is no $u\in V$, such that $\langle u,v\rangle\in E_{d}$ and, there is $w\in V$, such that $\langle w,v\rangle\in E_{A}$ then $M_{\vdash}^{\mathcal{C}}(v,i)=\{\ell(v)\}\vdash \ell(v)$, for each $i$, such that there is $\langle w,v\rangle\in E_{A}$ with $\delta(\langle w,v\rangle)=i$, and $M_{\vdash}^{\mathcal{C}}(v,j)=\symbvdash$, for every $j\in\mathbb{N}$, such that, there is no  $\langle w,v\rangle\in E_{A}$ with $\delta(\langle w,v\rangle)=j$,  and;
  \item [{\color{blue} Deductive Internal Node, no Ancestrality}] If $v\in V$ and, there is $u\in V$, such that $\langle u,v\rangle\in E_{d}$ and, there is no $w\in V$, such that $\langle w,v\rangle\in E_{A}$ then we have two cases:
    \begin{enumerate}
    \item There are only two $u_1,u_2\in V$, such that, $\langle u_i,v\rangle\in E_{d}$, for $i=1,2$, $\ell(u_1)=\delta_1$, $\ell(u_2)=\delta_1\imply \delta_2$ and, $\ell(v)=\delta_2$. Moreover, let  $I_i=\{j: M_{\vdash}^{\mathcal{C}}(u_i,j)\neq\symbvdash\}$. Thus, we have that:
      \begin{description}
      \item[If] $I_1=I_2$ {\bf then} for every $j\in I_1=I_2$, we have that:
        \[
        M_{\vdash}^{\mathcal{C}}(v,j)=\left\{\begin{array}{ll} \Delta_1\cup\Delta_2\vdash\delta_2  & \mbox{if $M_{\vdash}^{\mathcal{C}}(u_i,j)=\Delta_i\vdash\ell(u_i)$} \\ \symbvdash & \mbox{otherwise}\end{array}\right.
        \]
        , and every $j\in \mathbb{N}-(I_1\cup I_2)$, $M_{\vdash}^{\mathcal{C}}(v,j)=\symbvdash$. Moreover, if $L(\langle u_i,v\rangle)\downarrow$, $i=1,2$, then $L(\langle u_i,v\rangle)=Ant_{\mathcal{O}_{\alpha}}(M_{\vdash}^{\mathcal{C}}(u_i))$,  and; 
      \item[If] $I_1\neq I_2$ {\bf then} for every $j\in\mathbb{N}$, $M_{\vdash}^{\mathcal{C}}(v,j)=\symbvdash$
      \end{description}

    \item There is only one $u\in V$, such that, $\langle u,v\rangle\in E_{d}$, $\ell(u)=\delta_2$ and $\ell(v)=\delta_1\imply \delta_2$. Moreover, let  $I=\{j: M_{\vdash}^{\mathcal{C}}(u,j)\neq\symbvdash\}$. If $I\neq\emptyset$, and hence, for every $j\in I$, we have that:
        \[
        M_{\vdash}^{\mathcal{C}}(v,j)=\left\{\begin{array}{ll} \Delta-\{\delta_1\}\vdash\delta_1\imply\delta_2  & \mbox{if $M_{\vdash}^{\mathcal{C}}(u,j)=\Delta\vdash\ell(u)$ and $\ell(u)=\delta_2$} \\ \symbvdash & \mbox{otherwise}\end{array}\right.
        \]
        , and every $j\in \mathbb{N}-I$, $M_{\vdash}^{\mathcal{C}}(v,j)=\symbvdash$, and; For every $j\in\mathbb{N}$, $M_{\vdash}^{\mathcal{C}}(v,j)=\symbvdash$. Moreover, if $L(\langle u,v\rangle)\downarrow$ then $L(\langle u,v\rangle)=Ant_{\mathcal{O}_{\alpha}}^{\mathcal{C}}(M_{\vdash}(u))$. If $I=\emptyset$ then $M_{\vdash}^{\mathcal{C}}(v,j)=\symbvdash$, for every $j\in\mathbb{N}$.
      \end{enumerate}
  \item [{\color{blue} Deductive Internal Node with Ancestrality}] If $v\in V$ and, there is $u\in V$, such that $\langle u,v\rangle\in E_{d}$ and, there is  $w\in V$, such that $\langle w,v\rangle\in E_{A}$ then we have two cases:
    \begin{enumerate}
    \item There are only two $u_1,u_2\in V$, such that, $\langle u_i,v\rangle\in E_{d}$, for $i=1,2$, $\ell(u_1)=\delta_1$, $\ell(u_2)=\delta_1\imply \delta_2$ and, $\ell(v)=\delta_2$. Moreover, let  $I_i=\{j: M_{\vdash}^{\mathcal{C}}(u_i,j)\neq\symbvdash\}$. Thus, we have that:
      \begin{description}
      \item[If] $I_1=I_2$ {\bf then} for every $j\in I_1=I_2$, we have that:
        \[
        M_{\vdash}^{\mathcal{C}}(v,j)=\left\{\begin{array}{ll} \Delta_1\cup\Delta_2\vdash\delta_2  & \mbox{if $M_{\vdash}^{\mathcal{C}}(u_i,j)=\Delta_i\vdash\ell(u_i)$} \\ \symbvdash & \mbox{otherwise}\end{array}\right.
        \]
        , and every $j\in \mathbb{N}-(I_1\cup I_2)$, $M_{\vdash}^{\mathcal{C}}(v,j)=\symbvdash$, and;
      \item[If] $I_1\neq I_2$ {\bf then} for every $j\in\mathbb{N}$, $M_{\vdash}^{\mathcal{C}}(v,j)=\symbvdash$
      \end{description}

    \item There is only one $u\in V$, such that, $\langle u,v\rangle\in E_{d}$, $\ell(u)=\delta_2$ and $\ell(v)=\delta_1\imply \delta_2$. Moreover, let  $I=\{j: M_{\vdash}^{\mathcal{C}}(u,j)\neq\symbvdash\}$. If $I\neq\emptyset$, and hence, for every $j\in I$, we have that:
        \[
        M_{\vdash}^{\mathcal{C}}(v,j)=\left\{\begin{array}{ll} \Delta-\{\delta_1\}\vdash\delta_1\imply\delta_2  & \mbox{if $M_{\vdash}^{\mathcal{C}}(u,j)=\Delta\vdash\ell(u)$ and $\ell(u)=\delta_2$} \\ \symbvdash & \mbox{otherwise}\end{array}\right.
        \]
        , and every $j\in \mathbb{N}-I$, $M_{\vdash}^{\mathcal{C}}(v,j)=\symbvdash$, and; For every $j\in\mathbb{N}$, $M_{\vdash}^{\mathcal{C}}(v,j)=\symbvdash$. If $I=\emptyset$ then $M_{\vdash}^{\mathcal{C}}(v,j)=\symbvdash$, for every $j\in\mathbb{N}$.
    \end{enumerate}
    \item [{\color{blue} Divergent Deductive Internal Node }] In this case $v$ should not be the target of an ancestrality edge, otherwise $\mathcal{C}$ is not a valid {\bf pre} rDagProof and $M_{\vdash}^{\mathcal{C}}(v,j)=\symbvdash$, for every $j\in \mathbb{N}$. Moreover, the set $S=\{\langle w:\mbox{$\langle v,w\rangle\in E_{d}$}\}$ has at least two nodes\footnote{This is just the case for divergent deductive nodes}. We have two cases to consider:
    \begin{enumerate}
    \item There are only two $u_1,u_2\in V$, such that, $\langle u_i,v\rangle\in E_{d}$, for $i=1,2$, $\ell(u_1)=\delta_1$, $\ell(u_2)=\delta_1\imply \delta_2$ and, $\ell(v)=\delta_2$. Moreover, let  $I_i=\{j: M_{\vdash}^{\mathcal{C}}(u_i,j)\neq\symbvdash\}$ and $T=\{\rho(\langle v,w\rangle):w\in S\}$. Thus, we have that:
      \begin{description}
      \item[If] $I_1=I_2=I$ {\bf then} for every $j\in I_1=I_2=T$, we have that:
        \[
        M_{\vdash}^{\mathcal{C}}(v,j)=\left\{\begin{array}{ll} \Delta_1\cup\Delta_2\vdash\delta_2  & \mbox{if $M_{\vdash}^{\mathcal{C}}(u_i,j)=\Delta_i\vdash\ell(u_i)$} \\ \symbvdash & \mbox{otherwise}\end{array}\right.
        \]
        , and every $j\in \mathbb{N}-(I_1\cup I_2)$, $M_{\vdash}^{\mathcal{C}}(v,j)=\symbvdash$, and;
      \item[If] $I_1\neq I_2$ or $I_i\neq T$, $i=1$ or $i=2$ {\bf then} for every $j\in\mathbb{N}$, $M_{\vdash}^{\mathcal{C}}(v,j)=\symbvdash$
      \end{description}
    \item There is only one $u\in V$, such that, $\langle u,v\rangle\in E_{d}$, $\ell(u)=\delta_2$ and $\ell(v)=\delta_1\imply \delta_2$. Moreover, let  $I=\{j: M_{\vdash}^{\mathcal{C}}(u,j)\neq\symbvdash\}$, $T=\{\rho(\langle v,w\rangle):w\in S\}$. If $T=I\neq\emptyset$ then for every $j\in I$, we have that:
        \[
        M_{\vdash}^{\mathcal{C}}(v,j)=\left\{\begin{array}{ll} \Delta-\{\delta_1\}\vdash\delta_1\imply\delta_2  & \mbox{if $M_{\vdash}^{\mathcal{C}}(u,j)=\Delta\vdash\ell(u)$ and $\ell(u)=\delta_2$} \\ \symbvdash & \mbox{otherwise}\end{array}\right.
        \]
        , and every $j\in \mathbb{N}-I$, $M_{\vdash}^{\mathcal{C}}(v,j)=\symbvdash$, and; For every $j\in\mathbb{N}$, $M_{\vdash}^{\mathcal{C}}(v,j)=\symbvdash$. Moreover, if $L(\langle u,v\rangle)\downarrow$ then $L(\langle u,v\rangle)=Ant_{\mathcal{O}_{\alpha}}(M_{\vdash}(u))$. If $I=\emptyset$ then $M_{\vdash}^{\mathcal{C}}(v,j)=\symbvdash$, for every $j\in\mathbb{N}$.
    \end{enumerate}
  \item[{\color{blue} Target of Divergent Deductive Internal Node }] In this case $v$ is such that there is a divergent deductive node $u$ with $\langle u,v\rangle\in E_{d}$ and $\rho( \langle u,v\rangle)$ is defined. Thus, we have two cases:
    \begin{description}
    \item[$v$ is not target of an ancestrality edge] $M_{\vdash}^{\mathcal{C}}(v,0)=M_{\vdash}^{\mathcal{C}}(u,\rho(( \langle u,v\rangle)$ and $M_{\vdash}^{\mathcal{C}}(v,j)=\symbvdash$, for every $j\neq 0$;
    \item[$v$ is target of an ancestrality edge] There is $\langle w,v\rangle\in E_A$. Hence we set
      $M_{\vdash}^{\mathcal{C}}(v,\delta(\langle w,v\rangle))=M_{\vdash}^{\mathcal{C}}(u,\rho( \langle u,v\rangle)$ and $M_{\vdash}^{\mathcal{C}}(v,j)=\symbvdash$, for every $j\neq \delta(\langle w,v\rangle)$;
      \end{description}

  \end{description}

\end{definition}

{\bf Obs:} In what follows, sometimes we  use the notation $M_{\vdash}(v,i)$ instead of $M_{\vdash}^{\mathcal{C}}(v,i)$, whenever $\mathcal{C}$ can be easily infered from the context.  

A full subgraph of a graph $A=\langle V_{A}, E_{A}\rangle$ is  any graph $B=\langle V_{A}, E\rangle$, with $E\subseteq E_{A}$. The labelled version of full subgraph keeps all labels that label the elements of $V_{A}$ and $E$ with the same value they have in $A$.

\begin{definition}[Underlying-deductive-structure of an rDagProof]
  Given a {\bf pre} rDagProof $\mathcal{C}$ = $\langle V , E_{d}, E_{A}, r, \ell, L,\rho,\delta,\mathcal{O}_{\alpha}\rangle$. The full sub-graph of $\mathcal{C}$ when we consider all and  only all of the edges in $E_{d}$ is denoted by $\mathcal{C}|_{E_d}$. It is called the underlying deductive strutucture of $\mathcal{C}$.
  \end{definition}

\begin{definition}[maximal-path]
  Given a {\bf pre} rDagProof $\mathcal{C}$ = $\langle V , E_{d}, E_{A}, r, \ell, L,\rho,\delta,\mathcal{O}_{\alpha}\rangle$ and $v_k,\ldots,v_1$, $v_i\in V$, such that, $\langle v_{i+1},v_{i}\rangle\in E_{d}$, $i=1,\ldots,k-1$. We say that $v_1,\ldots,v_k$ is a maximal path in $\mathcal{C}|_{E_{d}}$, if and only if, $v_k$ is a top-formula. We say that the maximal-path starts in $v_1$.
\end{definition}

The lenght of the sequence of nodes $v_k,\ldots,v_1$, is $k$.

\begin{definition}[reverse-deductive height]\label{def:rdh}
  Given {\bf pre} r-DagProof $\mathcal{C}$ = $\langle V , E_{d}, E_{A}, r, \ell, L,\rho,\delta,\mathcal{O}_{\alpha}\rangle$ of $\alpha$. Let $\mathcal{C}|_{E_d}$ the sub-graph of $\mathcal{C}$ restricted to deductive edges only ($E_{d}$). The {\bf reverse deductive height} of a node $v\in V$ in $\mathcal{C}|_{E_d}$, named $rdh(v)$ is defined as:
  \[
  rdh(v)=max\{k: \mbox{$v_1,\ldots,v_k$ is a maximal-path with $v_1=v$}\}
  \]
  \end{definition}

The above definition~\ref{def:LocalEntailment} of $M_{\vdash}^{\mathcal{C}}$ is recursive. Given a {\bf pre} rDagProof $\mathcal{C}$, we can assign to each node $v\in V_{\mathcal{C}}$ the value of $rhd(v)$. By the recursion theorem, from set theory, we have that the function $M_{\vdash}^{\mathcal{C}}$ is well-defined for every node $v$ and natural number $i$  in  any rDagProof $\mathcal{C}$. Acoording to this assignment of values, we have that the value assigned to the root of $\mathcal{C}|_{E_d}$ is $h(\mathcal{C})$, the value of all of its leaves is 0 (zero) and the value of the children of any node is smaller than the value of their respective parent. Thus, $M_{\vdash}^{\mathcal{C}}$ is well-defined and unique for any $\mathcal{C}$.

We note the following well-known facts, regarding usual Kripke semantics for \mil \cite{MinimalCompleteness}, denoted by $\models_{\mil}$.

\begin{fact}[Soundness of ND \mil rules]\label{fact:NDsoundness}
  Consider $\Delta_1$ and $\Delta_2$ two sets of \mil formulas, and, $\delta_1$ and $\delta_2$ two \mil formulas. We have that:
  \begin{enumerate}
\item If $\Delta_1\models_{\mil} \delta_1$ and $\Delta_2\models_{\mil} \delta_1\imply\delta_2$ then $\Delta_1\cup \Delta_2\models_{\mil} \delta_2$, and;
\item\label{descarte} If $\Delta_1\models_{\mil}\delta_2$ then $\Delta_1-\{\delta_1\}\models_{\mil} \delta_1\imply\delta_2$
  \end{enumerate}
  Note that the above items are just the $\imply$-Intro and $\imply$-Elim rules. Concerning Item~\ref{descarte}, we have both cases $\delta_1\in\Delta_1$ and $\delta_1\not\in\Delta_1$, as it is the case with the $\imply$-Intro rule.
\end{fact}

In what follows we omit the symbol $\mil$ in the notation $\models_{\mil}$ whenever its meaning as the {\bf minimal entailment} is made clear.  

\begin{definition}[rDagProof correctness]\label{def:SoundrDag}
  Let $\mathcal{C}$ = $\langle V , E_{d}, E_{A}, r, l, L,\rho,\delta,\mathcal{O}_{\alpha}\rangle$ be a {\bf pre} r-DagProof. We say that $\mathcal{C}$ is correct iff $M_{\vdash}(r,0)\neq\symbvdash$ and, for each $i\neq 0$, $i\in\mathbb{N}$, $M_{\vdash}(r,i)=\symbvdash$.
\end{definition}

When a {\bf pre} rDagProof $\mathcal{C}$ is correct we simply call it rDagProof. Given a correct rDagProof $\mathcal{C}$ , such that, $M_{\vdash}^{\mathcal{C}}(r_{\mathcal{C}},0)=\Delta\vdash\beta$, we have that $\mathcal{C}$ is a certificate that $\beta$ as logical consequence of $\Delta$ in \mil. This is what we  state in Lemma~\ref{lemma:soundness} below.

\begin{lemma}[Local entailment sounds]\label{lemma:soundness}
  Let $\mathcal{C}$ = $\langle V , E_{d}, E_{A}, r, l, L,\rho,\delta,\mathcal{O}_{\alpha}\rangle$ be a correct rDagProof of $\alpha$. Thus, for every $v\in V$, for every $i\in \mathbb{N}$, such that $M_{\vdash}^{\mathcal{C}}(v,i)\neq\vdash$, then, if $M_{\vdash}^{\mathcal{C}}(v,i)=\Delta\vdash\ell(v)$, $\ell(v)=\beta$, then $\Delta\models_{\mil}\beta$. 
\end{lemma}

\begin{proof}
  By induction on the definition of $M_{\vdash}$ definition and using Lemma~\ref{fact:NDsoundness}.
\end{proof}

\begin{corollary}\label{coro:CoroLemmaSoundness}
  If in the stating of Lemma~\ref{lemma:soundness} above, we consider that:
\begin{itemize}
\item There is a formula $\delta$, subformula of $\alpha$, such that,  it is top-formula in $\mathcal{C}$ and,  there is no deductive path from this top-formula to the root $r$ of $\mathcal{C}$ that applies an $\imply$ introduction rule having $\delta$ as the antecedent of the formula that it is the conclusion of this application.
\end{itemize}
Then $M_{\vdash}^{\mathcal{C}}(v,i)=\Delta\vdash\ell(v)$, $\ell(v)=\beta$, with $\delta\in\Delta$.
\end{corollary}

\begin{proof}
  This corollary is a consequence of the proof of Lemma~\ref{lemma:soundness}. Its  proof is an extension of the induction proof of the lemma, by the inclusion the condition in the statement and verify that the formula $\delta$ is not removed during the evaluation of $M_{\vdash}^{\mathcal{C}}(v, i)$. Since it is a top formula, $\delta$ is included in the local entailment antecedent, in the basic step of the induction and, we do not remove it anymore.
\end{proof} 

\begin{theorem}[Completeness of rDagProofs]\label{completeness}
  For any \mil formula $\alpha$ and set of subformulas $\Delta$ of $\alpha$ and subformula $\beta$ of $\alpha$, we have that if $\Delta\models_{\vdash}\beta$ holds then there is a correct rDagProof $\mathcal{C}$, such that, $M_{\vdash}^{\mathcal{C}}(r_{\mathcal{C}},0)=\Delta^{\prime}\vdash\beta$, with $\Delta^{\prime}\subseteq\Delta$.
\end{theorem}

\begin{proof}
The system of Natural Deduction for \mil is sound and complete regarded the usual Kripke semantics for \mil \cite{MinimalCompleteness}. Since $ND_{\mil}$ proofs are particular cases of rDagProofs we have completeness of rDagProofs. Thus, if $\Delta\models\beta$ then there is a derivation $\Pi$ having $\beta$ as conclusion and a set $\Delta^{\prime}$ of open assumptions, with $\Delta^{\prime}\subseteq\Delta$. Taking $\Pi$ as a rDagProof and by Corollary~\ref{coro:CoroLemmaSoundness} we have that $M_{\vdash}^{\mathcal{C}}(r_{\mathcal{C}},0)=\Delta^{\prime}\vdash\beta$.
\end{proof}

\begin{theorem}[Soundness of rDagProofs]
  If $\mathcal{C}$ is a correct rDagProof and $M_{\vdash}^{\mathcal{C}}(r,0)=\Delta\vdash\beta$ then $\Delta\models_{\vdash}\beta$.
  \end{theorem}

\begin{proof}
This theorem is an immediate consequence of Lemma~\ref{lemma:soundness}
\end{proof}

\begin{corollary}
  If $\mathcal{C}$ is a correct rDagProof of $\alpha$ then $\alpha$ is a \mil tautology
  \end{corollary}

We have the following lemmata that help us to show that a correct r-DagProof, compressed by the technique that algorithm~\ref{algo:CompressEmNDProof} implements, is sound. Moreover, in the next section, we show an algorithm that checks whether a {\bf pre} rDagProof is correct or not. This verification is efficient (linear) on the size of the {\bf pre} rDagProof. 

In the next section, we show that for any Natural Deduction $\Pi$ that has its height linearly bounded by the size of its conclusion, $Compress(\Pi)$ is a correct rDagProof of $\alpha$. From this result, we can conclude that any \mil tautology has a succinct (polynomial) and correct rDagProof. In conclusion, we describe how to use this result to show that $NP=CoNP$.   

\section{The compression of linearly height-bounded rDagProofs preserves soundness} 

This section contains some lemmata that help us to prove that the correctness of rDagProofs are preserved by
the compression of rDagProofs.

%% \begin{lemma}\label{lemma:SoundCollapse}
%% Statement on Collapse

%% \end{lemma}

\begin{definition}[A-consistent sub-rDagProof]\label{def:FullSubrDagProof}
Let $\mathcal{D}$ be a {\bf pre} rDagProof of $\alpha$ and $\mathcal{C}$ a sub-rDagProof of $\mathcal{D}$. We say that $\mathcal{C}$ is an A-consistent sub-rDagProof of $\mathcal{D}$, if and only if, $\mathcal{D}\uparrow$ is graph-isomorphic to $\mathcal{C}$ and for every $\langle u,v\rangle\in E_{A}$, we have that $v\in V_{\mathcal{C}}$ and $u\not\in V_{\mathcal{C}}$, if and only if, $v\in V_{\mathcal{D}\uparrow k}$ and $u\not\in V_{\mathcal{D}\uparrow k}$.
  
  \end{definition}
\begin{lemma}[Local Entailment preservation under DetachLink]\label{lemma:Detach}
  Let $\mathcal{D}$ be a {\bf pre} rDagProof and $k$ a node of $\mathcal{D}$ that is the root of an instance of a A-consistent sub-rDagProof, cf. definition~\ref{def:FullSubrDagProof}, $\mathcal{C}$ of $\mathcal{D}$. Let $i\in\mathbb{N}$ be such that $i$ does not label any of the $E_d$ edges going out of $k$. Moreover, consider  $\mathcal{D}^{\prime}=DetachLink(\mathcal{D},k,\mathcal{C},i)$ as defined in Definition~\ref{def:Detach}. We have that for every $v\in V_{\mathcal{D}^{\prime}}$ and $j\in\mathbb{N}$, $j\neq i$, the following conditions hold:
  \begin{itemize}
  \item If $v\not\in V(\mathcal{D})\uparrow k$ then $M_{\vdash}^{\mathcal{D}}(v,j)=M_{\vdash}^{\mathcal{D}^{\prime}}(v,j)$, and;
  \item If $v\in V(\mathcal{D})\uparrow k$ then  $M_{\vdash}^{\mathcal{C}}(h^{-1}(v),j)=M_{\vdash}^{\mathcal{D}^{\prime}}(v,j)$, and;
    \item If $v\in V(\mathcal{D})\uparrow k$ then  $M_{\vdash}^{\mathcal{C}}(h^{-1}(v),0)=M_{\vdash}^{\mathcal{D}^{\prime}}(v,i)$; 
  \end{itemize}
  Where $h$ is the (full) labeled graph-isomorphism from $\mathcal{C}$ into $\mathcal{D}\uparrow k$.
\end{lemma}

\begin{proof}

  By inspecting  Definition~\ref{def:Detach}, we observe that  
  $\mathcal{D}^{\prime}=(\mathcal{D}-\mathcal{C}^{\prime})\cup \mathcal{C}$, where $\mathcal{D}\uparrow k=\mathcal{C}^{\prime}=h(\mathcal{C})$. We note that $M_{\vdash}^{\mathcal{D}^{\prime}}$ definition is then either on the complement of $\mathcal{D}\uparrow k$, or on $\mathcal{C}$. The former agrees with $M_{\vdash}^{\mathcal{D}}$ and the later agrees with $M_{\vdash}^{\mathcal{D}\uparrow k}=M_{\vdash}^{h(\mathcal{C})}$. Finaly, the third item in the statement of the lemma is related to the new edge that links the root of $\mathcal{C}$ to the former target of the $E_d$ edge that linked $k$ with this target. The top-formulas that are related to the basis steps of $M_{\vdash}$ recursion are accordingly accordingly associated to each respective part of $\mathcal{D}^{\prime}$, the same can be said about the recursive steps. So we have the desired result.  

  \end{proof}

\begin{corollary}[Soundness of DetachLink]\label{coro:DetachLink}
  Consider the conditions of Lemma~\ref{lemma:Detach}, above. We have that if $\mathcal{C}$ is correct and $\mathcal{D}$ is also correct then $DetachLink(\mathcal{D},k,\mathcal{C},i)$ is correct. Moreover, the local entailment is preserved, modulo the the isomorphism between $\mathcal{D}\uparrow k$ and $\mathcal{C}$.
  \end{corollary}

An immediate consequence of Corollary~\ref{coro:DetachLink} is that it the DetackLink operation can be repeated many times, without disturbing the soundness of the yielded rDagProof. Due to this, we have the following lemma.

\begin{lemma}[Soundness of Collapse]\label{lemma:CollapseSoundness}
  Let $\mathcal{D}$ and $\mathcal{C}$ be r-DagProofs. Let $\mathcal{C}$ be A-consistent subgraph of $\mathcal{D}$. Let $\mathcal{Y}$ be a list containing the roots of the instances of $\mathcal{C}$ in a fixed level $\mu$. We have that if $\mathcal{D}$ and $\mathcal{C}$ are correct then $Collapse(\mathcal{D},\mathcal{Y},\mathcal{C})$ is correct too. Moreover the local entailment is preserved as stated in Lemma~\ref{lemma:Detach}.
\end{lemma}

The above Lemma~\ref{lemma:CollapseSoundness} is the correctness proof of algorithm~\ref{algo:Collapse}.

The following theorem  proves the correctness of algorithm~\ref{algo:CompressEmNDProof}

\begin{theorem}[Soundness of Compressed rDagProofs]\label{lemma:SoundnessOfCompress}
  Let $\mathcal{D}=\langle V , E_{d}, E_{A}, r, l, L,\rho,\delta,\mathcal{O}_{\alpha}\rangle$ be a {\bf pre} rDagProof for a \mil formula $\alpha$. If $\mathcal{D}$ is correct then for every $3< p$ $Compress(\mathcal{D},p)$ is correct. Moreover, the local entailment is preserved, i.e., 
\end{theorem}

\begin{proof}

This proof proceeds by induction on the number of, recursive, calls to $Compress$. In the proof of the termination of algorithm~\ref{algo:CompressEmNDProof},  we have seen that it halts for every $\mathcal{D}$ and $p$, $3< p$. The basis step is trivial, and we use the inductive hypothesis together with Lemma~\ref{lemma:CollapseSoundness} to prove the inductive (recursive) step. 
  
\end{proof}

\section{On the complexity of verifying that a {\bf pre} rDagProof is correct or not}

The following algorithm performs a top-down sweeping in any {\bf pre} rDagProof to check whether it is correct or not. It is an iterative implementation of $M_{\vdash}^{\mathcal{C}}$ that prints can check whether the {\bf pre} rDagProof is correct, and certifies a tautology or not. In the case, it is not a tautology it prints ``DERIVATION''. Finally, it prints ``INCORRECT'' if the rDagProof is not correct. The definition of $M_{\vdash}$ points out the correctness of the algorithm. We analyse its computational complexity in the sequel. We have to note that the $update-and-check$, inside the iteration structures in lines~\ref{line:RegisterUpdate1} and~\ref{line:RegisterUpdate2} is responsible by updating the local entailment data-structure (the $Reg$ indexed structure) with $\symbvdash$ to indicate that the checking algorithm detected an incorrect {\bf pre} rDagProof.

 \begin{algorithm}[H]
%    \algsetup{linenosize=\tiny}
    \scriptsize
  \caption{Verifies whether a r-DagProof is valid}\label{algo:Check-rDagProof}
   \begin{algorithmic}[1]
     \STATE{{\bf Function}\;Check-rDagProof($\mathcal{C}$)}
     \FOR{$k=height(\mathcal{C}) downto 0$}\label{line:LevelSweep}
     \STATE{$L\leftarrow TopFormulas(k)$}
     \FOR{$top\in L$}\label{line:TopFormulas}
     \STATE{$Reg(top)\leftarrow\emptyset$}
     \FOR{$edge\in AncestorEdges(top)$}\label{line:AncestorTopFormulas}
     \STATE{$Reg(top)\leftarrow Reg(top)\cup\{\delta(edge)\mapsto \ell(top)\vdash\ell(top)\}$}
     \ENDFOR
     \STATE{$Reg(top)\leftarrow Reg(top)\cup\{0\mapsto \ell(top)\vdash\ell(top)\}$}\label{line:ZeroTopFormula}
     \ENDFOR
     \STATE{$I\leftarrow InternalNodes(k)$}
     \FOR{$v\in I$}\label{line:InternalNodes}
     \STATE{$q\leftarrow Premisses(v)$}
     \IF{$Divergent(v)$}
     \STATE{$Lv\leftarrow \{\langle v,w\rangle:\langle v,w\rangle\in E_A\}$}
     \IF{$\{i:Defined(Reg(q,i\}=\{\rho(e):e\in Lv\}$}
     \FOR{$i\in Reg(q)$}\label{line:RegisterUpdate1}
     \STATE{$Reg(v,i) \leftarrow update-and-check(v,q,i)$}
     \ENDFOR
     \ELSE
     \STATE{$Reg(v,0)\leftarrow\;\;\symbvdash$}
     \ENDIF
     \ELSE
     \IF{$TargetDivergent(v)$}\label{line:TargetDivergent}
     \STATE{$u\leftarrow\;\iota\{u:\langle u,v\rangle\in E_d\}$}
     \IF{$Defined(\rho(\langle u,v\rangle)\;\land\;\neg\exists\langle w,v\rangle\in E_A$}
     \STATE{$Reg(v,0)\leftarrow Reg(u,\rho(\langle u,v\rangle))$}
     \STATE{$Reg(v,j)\leftarrow\;\;\symbvdash\;\;\forall j\neq 0$}
     \ENDIF
     \IF{$Defined(\rho(\langle u,v\rangle)\;\land\;\exists\langle w,v\rangle\in E_A$}
     \STATE{$Reg(v,\delta(\langle w,v\rangle))\leftarrow\;Reg(u,\rho(\langle u,v\rangle))$}
     \STATE{$Reg(v,j)\leftarrow\;\;\symbvdash\;\;\forall j\neq\delta(\langle w,v\rangle)) $}
     \ENDIF
     \ELSE
     \FOR{$i\in Reg(q)$}\label{line:RegisterUpdate2}
     \STATE{$Reg(v,i) \leftarrow update-and-check(v,q,i)$}
     \ENDFOR
     \ENDIF
     \ENDIF
     \ENDFOR
     \ENDFOR
     \IF{$Reg(r)\neq\;\;\symbvdash$}
     \STATE{$ CORRECT $}
     \IF{$ Reg(r) ==\;\; \vdash\ell(r) $}
     \STATE{$TAUTOLOGY$}
     \ELSE
     \STATE{$ DERIVATION $}
     \ENDIF
     \ELSE
     \STATE{$ INCORRECT $}
     \ENDIF
%%     \REQUIRE{$3\leq p\in \mathbb{N}$, $\mathcal{T}$ is a {\bf EmND},  $h(T)\leq \size{c(\mathcal{T})}$}
%%     \ENSURE{a DAG $\mathcal{D}$ proving $c(\mathcal{T})$ of size smaller than $\size{c(\mathcal{T})}^{p-3}$}
%%     \IF{($\size{c(\mathcal{T})}^{p}< \size{\mathcal{T}})$}
%%     \STATE{$m_0\leftarrow MinLevel(Corollary(\mathcal{T}))$}
%%     \STATE{$\mathcal{D}\leftarrow \mathcal{T}$}
%%     \FOR{$k=m_0$ to $h(\mathcal{T})$}
%%     \STATE{$L\leftarrow SuperSubProofs(\mathcal{T},k)$}
%%     \FOR{$\langle\mathcal{Y},S_{\mathcal{Y}}\rangle \in L$}
%%     \STATE{$\langle r,\mathcal{D}_{\mathcal{Y}}\rangle\leftarrow Compress(\mathcal{Y})$}
%%     \STATE{$\mathcal{D}\leftarrow Collapse(\mathcal{D},S_{\mathcal{Y}},D_{\mathcal{Y}},r)$}
%%     \ENDFOR
%%     \ENDFOR
%% %    \STATE{FreeSlots $\leftarrow$ $\{v:\mbox{$v\in V(\mathcal{D})$ and $in(v)=0$}\}$}
%% %    \STATE{BoundSlots $\leftarrow$ $\{v:\mbox{$v\in V(\mathcal{D})$ and $\exists i\in\mathcal{O}_{B}$ $\langle r(D_{\mathcal{Y}}),v\rangle\in E_{D}^{i}$}\}$} 
%%     \STATE{$Return \langle r, \mathcal{D}\rangle$} 
%%     \ELSE
%%     \STATE{$Return \langle c(\mathcal{T}),\mathcal{T}\rangle$}
%%     \ENDIF
   \end{algorithmic}
\end{algorithm} 

 Let $\mathcal{C}=\langle V , E_{d}, E_{A}, r, l, L,\rho,\delta,\mathcal{O}_{\alpha}\rangle$ be a {\bf pre} rDagProof. We set $n_{v}=\size{V}$, $n_{A}=\size{E_{A}}$, $m=\size{\mathcal{O}_{\alpha}}=\size{T_{\alpha}}\leq len(\alpha)$ and $h=height(\mathcal{C})=height(\langle V,E_d\rangle)$.  In the sequel, all the line references are in algorithm~\ref{algo:Check-rDagProof}. We proceed to a worst case analysis to find an upper-bounded for the number of steps to check whether $\mathcal{C}$ is correct or not. The loop that starts in line~\ref{line:LevelSweep} consumes $h$ steps, for each of these steps, we have at most $n_v$ possible nodes that are top-formulas, this is what line~\ref{line:TopFormulas} sweeps using the ``for'' statement. Inside this ``for'' there is other nested iteration on the set of ancestor edges that target the top-formulas, and a consequent updating in the list of local entailments stored in the $Reg$ data-structure. Line~\ref{line:ZeroTopFormula} is responsible by the update of the main local-entailment (indexed by 0). After that, the loop that starts in line~\ref{line:InternalNodes} takes care of the internal nodes, and for each internal node, we have two possible cases, either it is a target of a divergent node and the needed udpate on the local indexed by the index 0 is made, or, the update of all indexes, including the 0, of the local entailment structure is made. The choice is made by the ``if'' statement in line~\ref{line:TargetDivergent}. The cost of with steps that takes care of the top-formulas is $h\times n_V\times n_A$. To this we have to add the cost of processing the internal nodes that is $h\times n_v\times n_A$ too, due to the cost of updating the indexes related to each Ancestor edge in $E_A$. However, as top-formulas and internal nodes are disjoint then we have that algorithm~\ref{algo:Check-rDagProof} when applied on $\mathcal{C}$ performs  the number of steps upper-bounded by disequality~\ref{steps:um}: 
 \begin{align}
 Steps(\mathcal{C})\leq h\times n_v\times n_A\label{steps:um} \\
 Steps(\mathcal{C})\leq n_v\times n_{v}^{3}=n_v^{4}\label{steps:dois}
  \end{align}
 Observing that $n_A\leq n_{v}^2$ and $h\leq n_v$ we have the upper-bound in disequality~\ref{steps:dois}
 Since we are counting steps, we can say that the time complexity to check whether a {\bf pre} rDagProof $\mathcal{C}$ is correct or not is polynomial, 4th power indeed,  on the size of $\mathcal{C}$.

 \section{A brief argument towards $CoNP=NP$}

 We have already discussed in the introduction of this text that when considering the complexity class $CoNP$, we are naturally limited to linearly height-bounded proofs. The proofs, in \mil,  of the non-hamiltonianicity of graphs, are linearly height bounded. See the appendix in ~\cite{Exponential} or ~\cite{Addendum} for a detailed explanation on this. If $NP\neq CoNP$ then the set of non-hamiltonian graphs there is no polynomially sized and verifiable in polynomial time certificate for each of its elements. Thus, the set $S$ of all formulas that have Normal Natural Deduction proofs linear height-bounded contains the valid for the non-hamiltonian graphs. Hence, by assuming that $NP\neq CoNP$, we have to conclude that $S$ is a family of normal super-polynomial proofs with linear height. If we consider any proof in $S$, either it is polynomially sized, and we have nothing to prove, or it is bigger than 
 $m^p$, for some $p>3$, where $m$ is the size of the proof's conclusion. We observe that the case $p\leq 3$ is subsumed by $p>3$, anyway. Thus, we can apply  Theorem~\ref{main} to show that this big proof is redundant, so we can apply the compression algorithm~\ref{algo:CompressEmNDProof} to obtain a correct rDagProof of size smaller than $m^p$, according to Lemma~\ref{lemma:SizeCompressed}. Finally, the algorithm~\ref{algo:Check-rDagProof} can check the correctness of this polynomially sized rDagProof in time upper-bounded by $m^{4p}$.

 The last paragraph provided a precise argumentation showing polynomial certificates for each non-hamiltonicity of each non-hamiltonian graph. We can check each of them is a (correct) certificate in polynomial time too. We can conclude that $CoNP\subseteq NP$, since non-hamiltonicity of graphs is a $CoNP$-complete problem. Having proved that $CoNP\subseteq NP$ we have proof that $NP=CoNP$, as the following reasoning shows, where $\overline{L}$ is the set-theoretical complement of $L$. We have used the logically simplest definition of the class  $CoNP$ class as $\{\overline{L}:L\in NP\}$.

  \[
  \begin{array}{c}
    \Rightarrow CoNP\subset NP \\
    \Downarrow \\
\mbox{$L\in NP$, iff, $\overline{L}\in CoNP$, ${\color{red}CoNP\subseteq NP}$ so $\overline{L}\in NP$, iff, $L\in CoNP$}\\
    \Downarrow \\
CoNP=NP
  \end{array}
  \]
 
More details and mathematical precision on this argument contain an alternative proof to the conjecture $CoNP=NP$ and are a matter for a further article.

\section{Conclusion}\label{sec:Conclusion}

This article shows that for any huge proof of a tautology in \mil we obtain a succinct certificate for its validity. Moreover,  we offer an algorithm able to check this validity in polynomial time on the certificate's size. We can use this result to provide a compression method to propositional proofs. Moreover, we can efficiently check the compressed proof without uncompressing it. Thus, we have many advantages over traditional compression methods based on strings. The compression ratio of techniques based on collapsing redundancies seems to be bigger, as shown in \cite{FlavioHaeuslerEBL} that reports some experiments with a variation of the Horizontal Compression method compared with Huffman compression. The second and more important advantage is the possibility to check for the validity of the compressed proof without having to uncompress it. In general, the original proof is huge, super-polynomial and hard check computationally.

Another application of the results in this article is to provide an alternative proof of $NP=CoNP$. In \cite{BoSL} we have a proof that $NP=NPSPACE$. An immediate consequence of this equality is that $NP=CoNP$. The approach that arises from the results we have shown here does not need Hudelmaier~\cite{Hudelmaier} linearly bounded sequent calculus for \mil logic. The proof reported in~\cite{BoSL}, on the other hand, needs Hudelmaier Sequent Calculus and a translation to Natural Deduction proofs that preserves the linear upper-bound. However, the resulted translation is not normal, and it is well-known that normalization does not preserve upper-bounds in general. Thus, we cannot apply our approach to the whole class of \mil tautologies to prove that $NPSPACE\subseteq NP$, for the use of normal proofs is essential to obtain the redundancy lemma, i.e., Lemma~\ref{main}. However, the compression method reported in this article, due to the redundancy lemma, provides knowledge to prove \mil short tautologies automatically. It seems easier than the use of the double certificate approach in~\cite{BoSL}.

\section{Acknowledgement}

We would like very much to thank professor Lew Gordeev for the work we have done together and the inspiration to follow this alternative approach. Thank Professor Luiz Carlos Pereira for his support, lessons and ideas on Proof Theory since the first course I have taken with him as a student. Thank the proof-theory group at Tuebingen-University, led by prof. Peter Schroeder-Heister. Many thanks to profs Gilles Dowek (INRIA) and Jean-Baptiste Joinet (univ. Lyon) for the intense interaction during this work's elaboration. Finally, we want to thank all students, former students, and colleagues who discussed with us in many stages during this work. We must have forgotten to mention someone, and we hope we can mend this memory failure in a nearer future.


\begin{thebibliography}{9}


\bibitem{Church} A.Church.
\newblock A set of postulates for the foundations of logic.
\newblock {\em Annals of Mathematics}, 33(2):346--366, 1932.
\newblock Corrections in \cite{Church1933}.

\bibitem{Church1933} A.Church.
\newblock A set of postulates for the foundations of logic (second paper).
\newblock {\em Annals of Mathematics}, 34(4):839--864, 1933.

\bibitem{Creszenci} Bovet, D. and Crescenzi, P.
  \newblock Introduction to the Theory of Complexity, 
  \newblock Prentice-Hall, 1994.

\bibitem{arora}
S.~Arora and B.~Barak.
\newblock {\em Computational Complexity: A Modern Approach}.
\newblock Cambridge University Press, 2009.
 
\bibitem{Gentzen1936} G.Gentzen.
\newblock Untersuchungen \"uber das logische schlie\ss{s}en {I}.
\newblock {\em Mathematische Zeitschrift}, 39:176--210, 1935.
\newblock English translation in \cite{Gentzen1964}.

\bibitem{Gentzen1964} G.Gentzen.
\newblock Investigations into logical deduction {I}.
\newblock {\em American Philosophical Quarterly}, 1(4):288--306, 1964.
\newblock English translation by M. E. Szabo.

\bibitem{Studia2019} L.Gordeev and E. H. Haeusler. \emph{Proof Compression and NP versus PSPACE}. Studia Logica 107, 53-83 (2019).

\bibitem{BoSL} L.Gordeev and E. H. Haeusler. \emph{Proof Compression and NP versus PSPACE II}. Bulletin of the Section of Logic, 2020, 18 pages, https://doi.org/10.18778/0138-0680.2020.16.

\bibitem{Addendum} L. Gordeev and E. H. Haeusler. \emph{Proof Compression and NP Versus PSPACE II: Addendum}, CoRR, abs/2011.09262, 2020, in https://arxiv.org/abs/2011.09262v2.  
   
 \bibitem{Haeusler2014} E. H. Haeusler. \emph{Propositional Logics Complexity and the Sub-Formula Property}, 
in Proceedings Tenth International Workshop on Developments in Computational
               Models, {DCM} 2014, Vienna, Austria, 13th July 2014. 

\bibitem{HaeuslerExp2Poly} Edward Hermann Haeusler. \emph{Every super-polynomial proof in purely implicational minimal logic
               has a polynomially sized proof in classical implicational propositional
               logic}, CoRR, abs/1505.06506, 2015, in http://arxiv.org/abs/1505.06506.

%%\bibitem{Exponential} E. H. Haeusler, \emph{On the Intrinsic Redundancy in Huge Natural Deduction proofs: Analysing Purely Implicational Minimal Logic Proofs}, Submitted to a Journal, 2020.
%%  shorturl.at/cvPX0

\bibitem{SuperPoly} E. H. Haeusler. \emph{On the Intrinsic Redundancy in Huge Natural Deduction proofs II: Analysing $M_{\imply}$ Super-Polynomial Proofs}, CoRR, submmited to Arxiv in August.

\bibitem{Exponential} E. H. Haeusler. \emph{Exponentially Huge Natural Deduction proofs are Redundant: Preliminary results on $M_\supset$}, CoRR, abs/2004.10659, 2020, in https://arxiv.org/abs/2004.10659.
 
\bibitem{exponential} E. H. Haeusler,
\newblock \emph{How Many Times do We Need an Assumption to Prove a Tautology in Minimal Logic? Examples on the Compression Power of Classical Reasoning}.
\newblock Electronic Notes in Theoretical Computer Science, v.315, 2015, pp 31-46.

\bibitem{Hudelmaier} Hudelmaier, J\"{o}rg. \emph{An $\mathcal{O}(n \log n)-Space$ Decision Procedure for Intuitionistic Propositional Logic}, Journal of Logic and Computation, volume 3, number 1, pp. 63-75, 1993. 

\bibitem{HilbertHis} D.\ Hilbert und W.\ Ackermann, {\it
Grundz\"uge der theoretischen Logik.} Julius Springer, Berlin 1928.
English version {\it Principles of Mathematical Logic}, Chelsea, New York
1950, ed.\ by R.\ E.\ Luce is based on 2nd German edition, 1938.   


\bibitem{Ladner} Ladner, Richard E.
\newblock \emph{The Computational Complexity of Provability in Systems of Modal Propositional Logic}.
\newblock SIAM J. Comput., n.3, v.6, 1977.

\bibitem{Prawitz}  D. Prawitz, \textbf{Natural deduction: a
  proof-theoretical study}. Almqvist \& Wiksell, 1965

\bibitem{LinSpeedUp-Space} J. Rothe, \textbf{Complexity Theory and Cryptology, Texts in Theoretical Computer Science, An EATCS Series},
Springer, 2005.


\bibitem{Savitch}  W. Savitch, \emph{Relationships between nondeterministic
  and deterministic tape complexities}, J. of Computer and System Sciences
  (4): 177--192 (1970)

\bibitem{MinimalCompleteness} Segerberg, Krister. \emph{Propositional Logics Related to Heyting's amd Johansson's}, Theoria, 34:26-61.

  \bibitem{FlavioHaeuslerEBL} José Fl\'{a}vio Cavalcante Barros Jr  and Edward Hermann Haeusler. \emph{A comparative study on proof compression techniques}, Brazilian Meeting on Logic, 2019, Proceedindgs, pages 85-86, in https://ebl2019.ci.ufpb.br/assets/Book\_of\_Abstracts\_EBL\_2019.pdf. In portuguese.

\bibitem{Statman79}  R. Statman, \emph{Intuitionistic propositional logic is
polynomial-space complete}, Theor. Comp. Sci. (9): 67--72 (1979)

\bibitem{Svejdar}  V. \^{S}vejdar, \emph{On the polynomial-space
completeness of intuitionistic propositional logic}, Archive for Math. Logic
(42): 711--716 (2003)
\end{thebibliography}
\end{document}